\documentclass[11pt,letterpaper]{article}
\usepackage[T1]{fontenc}
\usepackage[letterpaper,margin=1in]{geometry}
\usepackage{graphicx} % support the \includegraphics command and options
\usepackage{array} % for better arrays (eg matrices) in maths
\usepackage{setspace}
\usepackage{algorithm}
\usepackage{makecell}
\usepackage{tikz}

\usepackage{wrapfig}

\usepackage{amsmath, amssymb, amsfonts, verbatim}
\usepackage{amsthm}
\usepackage{hyphenat,epsfig,subcaption,multirow}
\usepackage{nicefrac}
\usepackage{paralist}
\usepackage{bbold}
\usepackage{mathrsfs}

\usepackage[font=footnotesize,labelfont=bf]{caption}

\usepackage[usenames,dvipsnames]{xcolor}

\DeclareFontFamily{U}{mathx}{\hyphenchar\font45}
\DeclareFontShape{U}{mathx}{m}{n}{
      <5> <6> <7> <8> <9> <10>
      <10.95> <12> <14.4> <17.28> <20.74> <24.88>
      mathx10
      }{}
\DeclareSymbolFont{mathx}{U}{mathx}{m}{n}
\DeclareMathSymbol{\bigtimes}{1}{mathx}{"91}

\usepackage{tcolorbox}
\tcbuselibrary{skins,breakable}
\tcbset{enhanced jigsaw}

\usepackage[normalem]{ulem}
\usepackage[compact]{titlesec}
\usepackage{thmtools}
\usepackage{thm-restate}

\definecolor{DarkRed}{rgb}{0.5,0.1,0.1}
\definecolor{DarkBlue}{rgb}{0.1,0.1,0.5}

\usepackage{nameref}
\definecolor{ForestGreen}{rgb}{0.1333,0.5451,0.1333}
\definecolor{Red}{rgb}{0.9,0,0}
\usepackage[linktocpage=true,
	pagebackref=true,colorlinks,
	urlcolor=black,
	linkcolor=DarkRed,citecolor=ForestGreen,
	bookmarks,bookmarksopen,bookmarksnumbered]
	{hyperref}
\usepackage[noabbrev,nameinlink]{cleveref}
\crefname{property}{property}{Property}
\creflabelformat{property}{(#1)#2#3}
\crefname{equation}{eq}{Eq}
\creflabelformat{equation}{(#1)#2#3}

\usepackage{bm}
\usepackage{url}
\usepackage{xspace}
\usepackage[mathscr]{euscript}

\usepackage{tikz}
\usetikzlibrary{arrows}
\usetikzlibrary{arrows.meta}
\usetikzlibrary{shapes}
\usetikzlibrary{backgrounds}
\usetikzlibrary{positioning}
\usetikzlibrary{decorations.markings}
\usetikzlibrary{patterns}
\usetikzlibrary{calc}
\usetikzlibrary{fit}
\usetikzlibrary{decorations}

\usepackage{mdframed}

\usepackage[noend]{algpseudocode}
\makeatletter
\def\BState{\State\hskip-\ALG@thistlm}
\makeatother

\usepackage[noadjust]{cite}
\usepackage{enumitem}

\usepackage[margin=1in]{geometry}

\newtheorem{lemma}{Lemma}[section]

\newtheorem{corollary}[lemma]{Corollary}
\newtheorem{theorem}[lemma]{Theorem}

\newtheorem{definition}[lemma]{Definition}

\newtheorem*{claim*}{Claim}
\newtheorem*{proposition*}{Proposition}
\newtheorem*{lemma*}{Lemma}
\newtheorem*{problem*}{Problem}

\crefname{lemma}{Lemma}{Lemmas}
\crefname{claim}{Claim}{Claims}

\newtheorem{mdresult}{Result}
\newenvironment{result}{
  \begin{mdframed}[backgroundcolor=lightgray!40,topline=false,rightline=false,leftline=false,bottomline=false,innertopmargin=2pt]
    \begin{mdresult}
}{
    \end{mdresult}
  \end{mdframed}
}

\newtheorem{remark}[lemma]{Remark}
\newtheorem{observation}[lemma]{Observation}

\DeclareMathOperator{\rk}{rk}
\DeclareMathOperator{\opt}{\textsc{opt}}
\DeclareMathOperator{\matroid}{\mathcal{M}}
\DeclareMathOperator{\groundset}{\mathcal{N}}
\DeclareMathOperator{\independentsets}{\mathcal{I}}
\newtheorem{property2}{Property}
\newtheorem{claim2}{Claim}
\usepackage{apptools}

\theoremstyle{definition}
\newtheorem{mdalg}{Algorithm}

\allowdisplaybreaks

\renewcommand{\qed}{\nobreak \ifvmode \relax \else
      \ifdim\lastskip<1.5em \hskip-\lastskip
      \hskip1.5em plus0em minus0.5em \fi \nobreak
      \vrule height0.75em width0.5em depth0.25em\fi}

\newcommand{\card}[1]{\left\vert{#1}\right\vert}

\newcommand{\set}[1]{\ensuremath{\left\{ #1 \right\}}}
\newcommand{\poly}{\mbox{\rm poly}}

\newenvironment{tbox}{\begin{tcolorbox}[
		enlarge top by=5pt,
		enlarge bottom by=5pt,
		 breakable,
		 boxsep=0pt,
                  left=4pt,
                  right=4pt,
                  top=10pt,
                  arc=0pt,
                  boxrule=1pt,toprule=1pt,
                  colback=white
                  ]%%
	}
{\end{tcolorbox}}

\newcommand{\II}{\ensuremath{\mathbb{I}}}

\newcommand{\mireal}[1][]{
  \ifx\relax#1\relax%
    \II(\mione \,; \mitwo)%
  \else%
    \II(\mione \,; \mitwo\mid #1)%
  \fi
}

\title{A Weighted-to-Unweighted Reduction for Matroid Intersection}
\date{}
\author{Aditi Dudeja\thanks{The Chinese University of Hong Kong, Shenzhen. \texttt{aditidudeja@cuhk.edu.cn}. This work was initiated while the author was affiliated with the University of Salzburg. This project has received funding from the European Research Council (ERC) under the European Union's Horizon 2020 research and innovation programme (grant agreement No 947702).} \and Mara Grilnberger\thanks{Department of Computer Science, University of Salzburg, Austria. \texttt{mara.grilnberger@plus.ac.at}. This work has been supported by the EXDIGIT (Excellence in Digital Sciences and Interdisciplinary Technologies) project, funded by Land Salzburg under grant number 20204-WISS/263/6-6022.}}
\begin{document}
\singlespacing
\begin{titlepage}
\clearpage\maketitle
\begin{abstract}
Given two matroids $\matroid_1$ and $\matroid_2$ over the same ground set, the \emph{matroid intersection} problem is to find the maximum cardinality common independent set. In the weighted version of the problem, the goal is to find a maximum weight common independent set. It has been a matter of interest to find efficient approximation algorithms for this problem in various settings. In many of these models, there is a gap between the best known results for the unweighted and weighted versions. In this work, we address the question of closing this gap. Our main result is a \emph{reduction} which converts any $\alpha$-approximate unweighted matroid intersection algorithm into an $\alpha(1-\varepsilon)$-approximate weighted matroid intersection algorithm, while increasing the runtime of the algorithm by a $\log W$ factor, where $W$ is the aspect ratio. Our framework is versatile and translates to settings such as streaming and one-way communication complexity where matroid intersection is well-studied.  As a by-product of our techniques, we derive new results for weighted matroid intersection in these models.

\end{abstract}
\thispagestyle{empty}
\newpage
\tableofcontents
\thispagestyle{empty}
\end{titlepage}
\newpage

\section{Introduction}
Matroid Intersection is a classical problem in combinatorial optimization for which faster algorithms have been a recent topic of interest. It generalizes many graph problems such as bipartite matching, colorful spanning trees, arborescences, tree packing, and disjoint spanning trees.

A matroid $\matroid$ is a tuple $(\groundset,\independentsets)$. Here, the groundset $\groundset$ is a finite collection of elements and $\independentsets\subseteq 2^{\groundset}$ is a family of sets that are independent in $\matroid$. The collection $\independentsets$ satisfies certain properties: \begin{inparaenum}[(i)]
    \item the empty set is independent \item every subset of an independent set is independent ({\bf downward closure}) \item if $A,B\in \mathcal{I}$ such that $|A|>|B|$ then there exist $a\in A\setminus B$ such that $B\cup\{a\}\in \independentsets$ ({\bf exchange property})
\end{inparaenum} The size of a maximum cardinality independent set is called the \emph{rank} of $\matroid$. In general, the rank also extends to arbitrary $S\subseteq \groundset$ and gives the size of the largest independent subset of $S$. Examples of matroids include the graphic matroid which is the collection of all spanning forests of a graph and the linear matroid which is the collection of all bases of a vector space.

\paragraph{Matroid Intersection.} In the problem of matroid intersection, we are given two matroids $\matroid_1=(\groundset,\independentsets_1)$ and $\matroid_2=(\groundset,\independentsets_2)$ with ranks $r_1$ and $r_2$. The goal is to obtain the \emph{maximum cardinality common independent set} of $\matroid_1$ and $\matroid_2$, which is the maximum-sized set independent in both $\matroid_1,\matroid_2$ (we denote its size by $r$). In this paper, we consider the weighted matroid intersection problem. Here, the input consists of two matroids $\matroid'_1=(\groundset',\independentsets'_1,w)$ and $\matroid'_2=(\groundset',\independentsets'_2,w)$ which are equipped with a weight function $w:\groundset \rightarrow \mathbb{R}^{> 0}$ and the goal is to compute the \emph{maximum weight common independent set} of $\matroid'_1$ and $\matroid'_2$. Note that, in general, the intersection of two matroids does not result in a matroid.

\paragraph{Accessing a Matroid.} Since a matroid in general requires exponentially (in $n$) many bits to describe, in many models of interest for this paper, the algorithms for matroid intersection assume oracle access to the matroids. For example, in the \emph{independence oracle} model, the algorithm can query if a set $S\subseteq \groundset$ is independent or not. In the \emph{rank oracle} model, the algorithm can query the rank of a set $S$. In the classical setting, the efficiency of the algorithm is measured by the number of oracle queries made. 

This paper concerns approximation algorithms for matroid intersection, which has been studied in the classical, parallel, streaming, and communication complexity settings. Steady progress (\cite{HuangKK16,ChekuriQ16,ChakrabartyLS0W19,Blikstad21}) have been made in the classical setting, in the $(1-\varepsilon)$-approximation regime, culminating in the works of \cite{BlikstadT25} and \cite{Quanrud24}. The former give an algorithm that computes a $(1-\varepsilon)$-approximate maximum cardinality independent set using $O(\frac{n\log n}{\varepsilon}+\frac{r\log^3 n}{\varepsilon^5})$ independence queries. The latter gave an algorithm that computes a $(1-\varepsilon)$-approximate maximum weight independent set using $O(\frac{n\log n}{\varepsilon}+\frac{r^{1.5}}{\varepsilon^4})$ independence queries. Beyond the classical setting, \cite{HuangS24,Terao24} studied the cardinality version of the matroid intersection problem in streaming and communication complexity settings (see \Cref{sec:applications} for a formal definition of these models). To the best of our knowledge, there are no weighted analogues of these results. 

The above discussion suggests that there is a gap between the best known results for weighted and unweighted versions of approximate matroid intersection. The question of closing this gap was first addressed by \cite{HuangKK16}, who showed that one can reduce the problem of $(1-\varepsilon)$-approximate maximum weight common independent set to the problem of solving $\frac{\log r}{\varepsilon}$ instances of maximum cardinality common independent set. \cite{CrouchS14} made progress on this question in the streaming setting. They showed that any $\alpha$-approximate maximum cardinality independent set algorithm can be converted into a $0.5\cdot (1-\varepsilon)\cdot\alpha$-approximate maximum weight independent set algorithm while preserving the pass complexity of the algorithm and incurring an additional $\log n$ factor in the space complexity. 

We continue this line of inquiry and obtain the following result (see \Cref{lem:staticreduction}, \Cref{lem:streamingreduction} and \Cref{lem:reductioncommcomplexity} for formal statements).
\begin{result}\label{result:main}
Let $\varepsilon>0$ and $\gamma_{\varepsilon}=\lceil(1/\varepsilon)^{1/\varepsilon}\rceil$. Let $W$ be the weight ratio of the matroids. We have for query, streaming, and communication models:
\begin{enumerate}
\item Given an $\alpha$-approximate unweighted matroid intersection algorithm $\mathcal{A}_u$ that makes $T_u(n,r)$ independence (resp., rank) queries, then there is an $\alpha(1-\varepsilon)$-approximate algorithm $\mathcal{A}_w$ for the weighted version that makes $O( T_{u}(n\gamma_{\varepsilon},r\gamma_{\varepsilon})\cdot\log W)$ independence (resp., rank) queries. 
\item Given an $\alpha$-approximate unweighted matroid intersection streaming algorithm $\mathcal{A}_u$ that takes space $S_u(n,r)$ and $p$ passes, there is an $\alpha(1-\varepsilon)$-approximate streaming algorithm $\mathcal{A}_w$ for the weighted version that takes space $S_u(\gamma_{\varepsilon}n,\gamma_{\varepsilon}r)$ and $p$ passes. 
\item Given $\mathcal{P}_{u}$, an $\alpha$-approximate one-way communication protocol for unweighted matroid intersection with message complexity $M_u(n,r)$, there exists an $\alpha(1-\varepsilon)$-approximate one-way communication protocol $\mathcal{P}_w$ for weighted matroid intersection with message complexity $M_u(\gamma_{\varepsilon}n,\gamma_{\varepsilon}r)$.
\end{enumerate}
In all of the above if the unweighted algorithm/protocol is deterministic, then the same is true of the weighted algorithm/protocol. 
\end{result}
\cite{HuangS24} introduced the notion of robust sparsifiers for matroid intersection. With our techniques we can extend the same notion to the problem of weighted matroid intersection. Matroid intersection captures many graph problems, and its specific instances such as bipartite $b$-matching have been studied in other settings such as MPC, distributed black-board model, and parallel shared memory work-depth model. Our structural results apply to bipartite $b$-matching as well, and therefore improve on or match the state-of-the-art for the weighted version of this problem in these settings.~\footnote{We note that it is plausible that algorithms for weighted bipartite $b$-matching can be obtained by using other techniques. We still mention it since our result offers a unified perspective on variants of the matroid intersection problem.} We discuss these applications in \Cref{sec:applications} and summarize them in \Cref{table:summary of results}. Lastly, our result builds on the work of \cite{BernsteinDL21}, who gave such a reduction for the case of bipartite matching. Their reduction has subsequently been used as a subroutine in other matching algorithms \cite{HashemiW24,LiuKK23,HuangS23}. Given the generality of matroid intersection and our reduction, we hope that our techniques will also have future applications. 

\paragraph{Related Work and Dependence on $\varepsilon$.}
We discuss some other works pertaining to the matroid intersection problem. \cite{GargJS21} gave a $(0.5-\varepsilon)$-approximate one-pass semi-streaming algorithm for weighted matroid intersection. \cite{Terao24,Quanrud24} gave multi-pass semi-streaming algorithms for the unweighted and weighted problems, respectively. In the query model, \cite{ChekuriQ16} gave a scaling algorithm for weighted matroid intersection in the $(1-\varepsilon)$-regime which uses unweighted $(1-\varepsilon)$-approximate matroid intersection algorithms which satisfy some special properties (namely, \cite{Cunningham86,Blikstad21}) in each scale, incurring an overhead of $O(\varepsilon^{-2}\log \varepsilon^{-1})$.\footnote{To the best of our knowledge, \cite{ChekuriQ16} require that the unweighted matroid intersection algorithm in each scale compute a maximal set of augmenting paths. We discuss this further in \Cref{sec:futurework}.} Given the state of affairs, it is reasonable to ask if there is a cross-paradigm weighted-to-unweighted reduction for matroid intersection applicable for \emph{arbitrary} approximation ratios with a polynomial dependence on $\varepsilon^{-1}$. We discuss some barriers to this in the Future Work section (\Cref{sec:futurework}). 

Lastly, we note that while we focus on obtaining efficient approximation algorithms, there is a long line of work that considers the exact version of matroid intersection \cite{Edmonds03,AignerD71,Lawler75,Cunningham86,ChakrabartyLSSW19,Nguyen19,BlikstadBMN21,Blikstad21}. In several instances, approximation algorithms for both the unweighted and weighted problem are used to speed up the exact algorithms. For such an approach, having a small dependence on $\varepsilon$ in the number of queries is desirable. Our reduction on the other hand, incurs a large $\varepsilon$-dependence and is therefore not suitable for such applications.  

\section{Overview of Techniques}

Before we give the high-level overview of our techniques, we introduce some notation. We also mention the dual linear program for the matroid intersection problem, which will be crucial to our proofs. 
\paragraph{Notation.} Each matroid $\matroid=(\groundset,\independentsets)$ is associated with a rank function, $\rk:2^{\groundset}\rightarrow \mathbb{N}$. For $S\subseteq \groundset$, we use $\rk(S)$ to denote the size of the maximal independent subset of $S$. The rank of a matroid $\rk(\matroid)$ denotes the size of the maximal independent set in $\matroid$. For a matroid intersection instance $\matroid_1,\matroid_2$, we will often use $r_1,r_2$ as shorthand for $\rk(\matroid_1),\rk(\matroid_2)$, respectively. Moreover, let $r$ denote the rank of the intersection. We will use $\opt$ to denote the size of the maximum cardinality common independent set of $\matroid_1,\matroid_2$. For a set $S \subseteq \mathcal{N}$, the matroid $\mathcal{M}$ restricted to $S$ is defined as $\mathcal{M} | S := (S, \mathcal{I}|S)$ where $\mathcal{I}|S := \{I \subseteq S | I \in \mathcal{I}\}$. 

Throughout this paper, we will use ``prime'' to distinguish weighted matroids from unweighted ones. An unweighted matroid will be denoted by $\matroid$, while a weighted matroid will be referred to as $\matroid'$. The same is true for associated sets (such as $\groundset'$) and functions (such as $\rk'$). We let $W=\frac{\max_{e\in \mathcal{N}'}w(e)}{\min_{e\in \mathcal{N}'}w(e)}$.

\paragraph{Linear Programming for Matroid Intersection.} In this paper, we will use duality for some of our arguments. For the unweighted matroid intersection problem over $\mathcal{M}_1=(\mathcal{N},\mathcal{I}_1)$ and $\mathcal{M}_2=(\mathcal{N},\mathcal{I}_2)$ we give the standard linear program relaxation. Let $\rk_1$ and $\rk_2$ denote the rank functions associated with $\mathcal{M}_1$ and $\mathcal{M}_2$, respectively. Let $x: \groundset \rightarrow \mathbb{R}^{\geq 0}$ and $x(S) := \sum_{e \in S} x(e)$ for all $S \subseteq \groundset$, the linear program relaxation is given by the following:
\begin{equation}\label{eqn:primalunweighted}
\begin{array}{ll@{}ll}
\text{maximize}  & x(\groundset)&\\
\text{subject to}& \\
&x(S) \leq \rk_1(S) \text{ and } x(S) \leq \rk_2(S)\text{ for all } S \subseteq \mathcal{N}.
\end{array}
\end{equation}
For the linear program relaxation of the weighted problem with weight function $w : \groundset \rightarrow \mathbb{R}^{\geq 0}$, the sum $\sum_{e \in \groundset} x(e) w(e)$ is maximized under the same constraints.
Now, we give the dual linear programs for the relaxation of the matroid intersection problem as presented in \Cref{eqn:primalunweighted} over $\mathcal{M}_1$ and $\mathcal{M}_2$. Let $y,z: 2^{\mathcal{N}}\rightarrow \mathbb{R}^{\geq 0}$, the dual problem corresponds to:
\begin{equation}\label{eqn:dualunweighted}
\begin{array}{ll@{}ll}
\text{minimize}  & \displaystyle f(y,z)=\sum_{S \subseteq \mathcal{N}} y(S) \rk_{1}(S) + z(S) \rk_{2} (S)&\\
\text{subject to}& \\
&\sum_{S: e \in S} y(S) + z(S) \geq 1 \text{ for all } e \in \mathcal{N}
\end{array}
\end{equation}

Suppose $\mathcal{M}'_1=(\mathcal{N}',\mathcal{I}'_1,w)$ and $\mathcal{M}'_2=(\mathcal{N}',\mathcal{I}'_2,w)$ is an instance of weighted matroid intersection. Let $\rk'_1$ and $\rk'_2$ be the rank functions associated with them. Let $y',z':2^{\mathcal{N}}\rightarrow \mathbb{R}^{\geq 0}$, the associated dual linear program is defined as follows:

\begin{equation}\label{eqn:dualweighted}
\begin{array}{ll@{}ll}
\text{minimize}  & \displaystyle g(y',z')=\sum_{S \subseteq \mathcal{N}} y'(S) \rk'_{1}(S) + z'(S) \rk'_{2} (S)&\\
\text{subject to}& \\
&\sum_{S: e \in S} y'(S) + z'(S) \geq w(e) \text{ for all } e \in \mathcal{N}'
\end{array}
\end{equation}

We will also need the following result on the structure of duals (see \cite[Chapter 41]{Schrijver03} for a proof).

\begin{lemma}[\cite{Edmonds03}]\label{lem:chaineddualsexistence}
    We say that a sequence of sets $S_1,S_2,\cdots, S_l$ form a chain if $S_1\subsetneq S_2\subsetneq S_3\subsetneq\cdots\subsetneq S_l$. There exists optimal solutions $y,z$ for Linear Program \ref{eqn:dualunweighted} such that $\textup{supp}(y)$ forms a chain and $\textup{supp}(z)$ forms a chain. This also applies to Linear Program \ref{eqn:dualweighted}. Furthermore, the optimal solutions to Linear Program \ref{eqn:dualunweighted} are always integral. For Linear Program \ref{eqn:dualweighted} if $w$ is integral, then there exists optimal integral dual solutions $y,z$ such that $\text{supp}(y),\text{supp}(z)$ are chains.
\end{lemma}
\subsection{Our Framework}

We follow the overall framework of \cite{BernsteinDL21}, who obtained \Cref{result:main} for the specific case of bipartite matching. This framework has since been a useful subroutine many matching algorithms \cite{HuangS23,LiuKK23}. The reduction framework will have three steps. In the following $\opt$ will refer to the size of the maximum cardinality common independent set of two unweighted matroids and $\opt'$  will refer to the weight of the maximum weight common independent set of two weighted matroids.

\begin{enumerate}
    \item\label{step:one} {\bf Aspect Ratio Reduction:} Suppose we are given two arbitrary matroids $\matroid'_1=(\groundset',\independentsets'_1,w)$ and $\matroid'_2=(\groundset',\independentsets'_2,w)$ with $w:\groundset'\rightarrow \mathbb{R}^{> 0}$. In this step we show that we can transform $\matroid'_1,\matroid'_2$ into matroids whose groundset elements have weights $\set{1,2,\cdots, \gamma_{\varepsilon}}$ and the maximum weight common independent set of the new matroids has weight at least $(1-\varepsilon)\opt'$.
    \item\label{step:weightedtounweighted} {\bf Weighted-to-Unweighted Conversion:} Suppose we are given matroids $\matroid'_1=(\groundset',\independentsets'_1,w)$ and $\matroid'_2=(\groundset',\independentsets'_1,w)$ such that $w(e)\in \{1,2,\cdots, W\}$ are integers for all $e\in \groundset'$. In this step, we will convert $\matroid'_1,\matroid'_2$ into unweighted matroids $\matroid_1=(\groundset,\independentsets_1)$ and $\matroid_2=(\groundset,\independentsets_2)$ (a process we refer to as matroid unfolding). Let $\opt$ be the size of the maximum cardinality common independent set of $\matroid_1,\matroid_2$. Then, this reduction will ensure that $\opt'=\opt$. Additionally, $|\groundset|\leq W\cdot |\groundset'|$, $\opt\leq r\cdot W$, and $\rk(\matroid_i)\leq W\cdot \rk'(\matroid'_i)$ for $i\in \{1,2\}$. 
\end{enumerate}
Before we mention the final step, we note that despite being given weighted matroids $\matroid'_1,\matroid'_2$ with an arbitrary weight function, we can assume by Step \ref{step:one} that element weights are integers in $\{1,2,\cdots, \gamma_{\varepsilon}\}$. Combining this assumption with Step \ref{step:weightedtounweighted} gives us unweighted matroids $\matroid_1,\matroid_2$ that are only a $\gamma_{\varepsilon}$-factor larger in size. Thus, we can apply any $\alpha$-approximate unweighted matroid intersection algorithm on $\matroid_1,\matroid_2$ to obtain $I\in \independentsets_1\cap\independentsets_2$ such that $|I|\geq \alpha \cdot \opt\geq \alpha\cdot (1-\varepsilon)\cdot \opt'$. 
\begin{enumerate}
    \item[3.] {\bf Refolding:} At this point assume by the previous discussion that we have $I\in \independentsets_1\cap \independentsets_2$. We would like to ``reverse'' the operation of unfolding and ``refold'' $I$ back in the original matroids $\matroid'_1,\matroid'_2$. The refolding process outputs $I'\subseteq \groundset'$. Note that $I'$ may not necessarily be an independent set, but we are able to show that there exists $J'\subseteq I'$ such that $J'\in \independentsets'_1\cap \independentsets'_2$ and $w(J')\geq (1-\varepsilon)\cdot \alpha\cdot \opt'$. The key advantage is that $\card{I'}=O(\gamma_{\varepsilon}\cdot r)$. Intuitively, it seems more efficient to extract a maximum weight independent set from a small subset of $\groundset'$.
\end{enumerate}

Our aspect ratio reduction technique also extends to the matroid parity problem, as we argue in \Cref{sec:aspectratio}. Additionally, we note that while the overall framework we follow is that of \cite{BernsteinDL21}, implementing it requires new matroid specific arguments, which we give an overview of. 

\paragraph{Aspect-Ratio Reduction. }Our starting point is the work of \cite{GuptaP13}. At a high-level our transformation to reduce the aspect-ratio of matroids $\matroid'_1=(\groundset',\independentsets'_1,w)$ and $\matroid'_2=(\groundset',\independentsets'_2,w)$ (where $w:\groundset'\rightarrow \mathbb{R}^{> 0}$) will proceed as follows. We will split $\groundset'$ into groups with geometrically increasing weights, so that within a group the ratio between weights is $\varepsilon^{-1}$. Then, we will delete one group for every $\varepsilon^{-1}$ groups (we refer to these as \emph{missing groups}), and merge consecutive groups into \emph{weight classes}. Each weight class $L_{i}$ now has a weight ratio of $\gamma_{\varepsilon}$. Within each class we would like to maintain for each $i\geq 0$, a common independent set of $\matroid'_1\mid L_i$ and $\matroid'_2\mid L_i$ called $I'_i$. The idea is to show that a) there exist $I'\subseteq \cup_{i\geq 0}I'_{i}$ such that $I'\in \independentsets'_1\cap \independentsets'_2$ and $w(I')\geq (1-\varepsilon)\cdot \sum_{i\geq 0} w(I'_i)$  and b) we can extract $I'$ from $\cup_{i\geq 0 }I'_i$ efficiently. 

\cite{GuptaP13} do the above procedure for the case of matchings. Let $\set{M_i}_{i\geq 0}$ be the matchings. They combine $M_i$'s greedily (in descending order of weight classes) to obtain $M$. Every edge in a higher weight class can be “blamed” for excluding at most two edges from each of the lower weight classes. Since there is a gap of $\varepsilon^{-1}$ between two consecutive weight classes, an edge $e$ in a higher weight class only blocks out a total weight of $4\varepsilon\cdot w(e)$ from the lower weight classes. Thus, greedy aggregation leads to a loss of weight at most $4\varepsilon\cdot \sum_{i}w(M_i)$, and therefore, $w(M)\geq (1-\varepsilon)\cdot \sum_{i}w(M_i)$.

For the case of matroids we also combine $I'_i$'s greedily (in descending order of weight classes) to obtain $I'$. However, we cannot hope to make a ``local'' charging argument as the one above. Instead, we show a structural lemma (\Cref{lem_charging}) which suggests that if $I'$ excludes $l$ elements of $I'_j$, then there are  at least $\frac{l}{2}$ elements of $I'\setminus\cup_{f\leq j} I'_f$ which can be ``blamed'' for the excluded elements. Thus, on average, each element of a higher weight class gets blamed by at most two elements from each of the lower weight classes. With this more ``global'' argument we are able to recover the charging argument mentioned above and conclude that $w(I')\geq (1-4\varepsilon)\cdot \sum_{i}w(I'_i)$. Finally, one can show by a simple pigeonhole argument that there is a choice of \emph{missing groups} such that $\sum_{i}w(I'_i)\geq (1-\varepsilon)\cdot\opt'$ and therefore, $w(I')\geq (1-5\varepsilon)\cdot \opt'$.

\paragraph{Weighted-to-Unweighted Reduction. } Our starting point is the graph unfolding technique of \cite{KaoLST01}, who give a weighted-to-unweighted reduction for matchings in integer-weighted graphs. Given a weighted graph $G'$, they create an unweighted graph $G$ by replacing each edge $e\in G'$ with $w(e)$-many edges in $G$ such that $\mu(G)=\mu_{w}(G')$ (here $\mu(G)$ refers to the size of the maximum cardinality matching and $\mu_{w}(G')$ refers to the weight of the maximum weight matching of $G'$).

Suppose we are given two matroids $\matroid'_1=(\groundset',\independentsets'_1,w)$ and $\matroid'_2=(\groundset',\independentsets'_2,w)$ with integer weights and a maximum weight common independent set of weight $\opt'$. We create matroids $\matroid_1=(\groundset,\independentsets_1)$ and $\matroid_2=(\groundset,\independentsets_2)$, with maximum cardinality common independent set of size $\opt$, by replacing each element $e\in \groundset'$ with $w(e)$ copies of the element $\{e_{1},\cdots, e_{w(e)}\}$. However, we need to impose additional structure on $\matroid_1,\matroid_2$ so that \begin{inparaenum}[(a)]
    \item\label{item:opt} $\opt=\opt'$ \item\label{item:queries} We can implement independence (resp. rank) queries to $\matroid_1,\matroid_2$ using a small number of independence (resp. rank) queries to $\matroid'_1,\matroid'_2$.
\end{inparaenum} In order to achieve this, we examine the structure of the solution for the dual program of weighted matroid intersection mentioned in LP \ref{eqn:dualweighted}. To ensure $\opt\leq \opt'$, we assume that given an optimal dual solution $y', z'$, we want to create a suitable solution $y, z$ for \ref{eqn:dualunweighted}. To this end, we can ``split'' the value of $y'(S')$ (and $z'(S')$) for each $S' \subseteq \groundset'$ among a collection of sets $\mathcal{S}_y$ (and $\mathcal{S}_z$), such that $\sum_{S \in \mathcal{S}_y} y(S) \leq y'(S')$ and $\rk_1(S) \leq \rk'_1(S')$ for all $S \in \mathcal{S}_y$ (analogously $\rk_2(T) \leq \rk'_2(T')$ for all $T \in \mathcal{T}_z$). To achieve the latter, the independence-structure of $S$ in $\matroid_1$ should match the structure of $S'$ in $\matroid_1'$ (similar for $\matroid_2$). Simultaneously, the constraints of LP \ref{eqn:dualunweighted} need to be satisfied. We can achieve this if we can also ensure $\sum_{S \in \mathcal{S}_y} y(S) \geq y'(S')$ (similarly for $z'$). Intuitively, for each matroid $\matroid_1, \matroid_2$, we want to partition the new groundset $\groundset$ into $W$ sets $N_1, N_2, \ldots, N_W$, where each set contains at most one copy of any original element in $\groundset'$. To create the new matroids $\matroid_1, \matroid_2$, the copied elements in $N_i$ inherit the independence-structure of the corresponding original elements from $\matroid_1, \matroid_2$. Further, the structure given to elements in $N_i$ should not interfere with the structure imposed upon $N_j$ for all $j \neq i$, i.e. the circuits of the resulting matroid are contained in one set $N_i$ for $i \in [W]$. The remaining question is how to create such partitions of $\groundset$, which inform the new independence structures given by $\independentsets_1, \independentsets_2$. \Cref{lem:chaineddualsexistence} implies there exist optimal integral dual solutions $y', z'$ that form chains. That is, let $\text{supp}(y')=\set{S^1,\cdots, S^l}$ and $\text{supp}(z')=\set{T^1,\cdots, T^k}$, then:
\begin{itemize}
    \item $S^l\subsetneq S^{l-1}\subsetneq S^{l-2}\subsetneq\cdots\subsetneq S^1$, and
    \item $T^{k}\subsetneq T^{k-1}\subsetneq T^{k-2}\subsetneq\cdots\subsetneq T^1$
\end{itemize}

This lemma gives us some hint about the right definition for $\independentsets_1,\independentsets_2$. For each set $S' \in \text{supp}(y')$ and each $T' \in \text{supp}(z')$, we now describe the sets in $\mathcal{S}_y$ and $\mathcal{T}_z$ from the explanation above.  For a fixed $j\in [l]$ and  $f\in \set{\sum_{i=1}^{j-1} y'(S^{i})+1,\cdots, \sum_{i=1}^{j} y'(S^{i})}$ let $S^{j}_{f}:=\set{e_f\mid e\in S^j}$ and $y(S^j_f)=1$. Similarly, for $j\in [k]$, $f\in \set{\sum_{i=1}^{j-1} z'(T^{i})+1,\cdots, \sum_{i=1}^{j} z'(T^{i})}$ let $T^{j}_{f}:=\set{e_{w(e)-f+1}\mid e\in T^j}$ and $z(T^j_f)=1$. Hence, this informs the partitions of $\groundset$, that ``inherit'' the original independence-structure of $\matroid'_1$ and $\matroid'_2$ (\Cref{def:matroidunfolding}). For $\matroid_1$, the sets in the partition are formed by elements $\{e_j | e \in \groundset' \text{ and } j \leq w(e)\}$ for all $j  \in [W]$. Similarly, $\matroid_2$ is partitioned into sets $\{e_{w(e) - j + 1} | e \in \groundset' \text{ and } j \leq w(e)\}$ for all $j  \in [W]$, as suggested by the dual. An example of this conversion process can be seen in \Cref{fig:matroidUnfoldingFigure}. After some careful arithmetic, one can show that $f(y,z)=g(y',z')$ (\Cref{thm:equivalence}). Moreover, one can also show the feasibility of the duals $y,z$ (\Cref{lem:feasibilityyz}). In particular, we show that for each $e\in \groundset'$, $e_i$ is covered by either $y$ or $z$ for all $i\in [w(e)]$.

\begin{itemize}
    \item Let $j\in [l]$ be the index of the last set $S^j$ such that $e\in S^j$. By the chain property,  $\set{e_{1},e_2,\cdots, e_{\sum_{f\geq 1}^j y'(S^f)}}$ are covered by $y$. Similarly, let $p\in [k]$ be the index of the last set $T^p$ such that $e\in T^p$. By the chain property, $\set{e_{w(e)},e_{w(e)-1},\cdots, e_{w(e)-\sum_{f\geq 1}^pz'(S^f)+1}}$ is covered by $z$. By feasibility of $y',z'$ one can conclude that $\{e_1,\cdots, e_{w(e)}\}$ is covered by either $y$ or $z$.
\end{itemize}
This motivates our \Cref{def:matroidunfolding} for the matroids $\matroid_1,\matroid_2$, describing the new independent sets for the unweighted matroids created from the original weighted matroids $\matroid'_1,\matroid'_2$. Our main technical contribution for this part of the framework is to prove various lemmas that show the two goals (\ref{item:opt}) and (\ref{item:queries}) we mentioned above can be achieved.

\begin{figure}
    \centering
            \scalebox{0.7}{
			\begin{tikzpicture}[
				roundnode/.style={circle, draw=black, thin,minimum size=1mm},
				helpnode/.style={circle, draw=white, thin,minimum size=1mm},
			]
			
				\node[roundnode]  (A) at (0, 0) {};
				\node[roundnode]   (C) at (1.5, -1) {};
				\node[roundnode]   (D) at (0, -2) {};

				\node[roundnode]  (E) at (-0.5, -3.5) {};
				\node[roundnode]   (F) at (-0.5, -5.5) {};
				\node[roundnode]   (G) at (1.2, -3.5) {};
				\node[roundnode]   (H) at (1.2, -5.5) {};
				
				\draw[dashed] (-2, -2.75) -- (12,-2.75);
				\node[label=left:$\mathcal{M}'_1:$]   (i1) at (-1, 0.3) {};
				\node[label=left:$\mathcal{M}'_2:$]   (i2) at (-1, -3.2) {};

				\node[label=left:$\mathcal{M}_1:$]   (i3) at (4.5, 0.3) {};
				\node[label=left:$\mathcal{M}_2:$]   (i4) at (4.5, -3.2) {};

				%\draw[thick, color=blue] (A) -- (B);
				\draw[thick] (A) -- (C) node [midway, above=2pt] {$b, 1$};
				\draw[thick, blue] (C) -- (D) node [midway, below=2pt] {$c, 2$};
				\draw[thick, blue] (A) -- (D) node [midway, right] {$d, 2$};

				\path (A.west) edge [out=200,in=160] node [midway, left] {$a, 3$} (D.west);

				\path (E.west) edge [out=220,in=140] node [midway, left] {$a, 3$} (F.west);
				\path (G.east) edge [out=320,in=40] node [midway, right] {$b, 1$} (H.east);
				\draw[thick, blue] (E) -- (F) node [midway, right=2pt] {$c, 2$};
				\draw[thick, blue] (G) -- (H) node [midway, left=2pt] {$d, 2$};
				%\draw[thick, color=blue] (B) -- (D);
				
				\draw[thick, ->] (3,-1) -- (3.5, -1);
				\draw[thick, ->] (3,-4.5) -- (3.5, -4.5);

				\node[roundnode]  (A1) at (5, -0.2) {};
				\node[roundnode]   (C1) at (6.3, -1) {};
				\node[roundnode]   (D1) at (5, -1.8) {};

				\draw[thick, blue] (A1) -- (C1) node [midway, above=2pt] {$b_1$};
				\draw[thick] (C1) -- (D1) node [midway, below=2pt] {$c_1$};
				\draw[thick] (A1) -- (D1) node [midway, right] {$d_1$};
				\path[blue] (A1.west) edge [out=200,in=160] node [midway, left] {$a_1$} (D1.west);

				\node[roundnode]  (A2) at (8, -0.2) {};
				\node[roundnode]   (C2) at (9.4, -1) {};
				\node[roundnode]   (D2) at (8, -1.8) {};

				\draw[thick, blue] (C2) -- (D2) node [midway, below=2pt] {$c_2$};
				\draw[thick] (A2) -- (D2) node [midway, right] {$d_2$};
				\path[blue] (A2.west) edge [out=200,in=160] node [midway, left] {$a_2$} (D2.west);

				\node[roundnode]  (A3) at (11, -0.2) {};
				\node[roundnode]   (C3) at (11, -1.8) {};

				\path (A3.west) edge [out=200,in=160] node [midway, left] {$a_3$} (C3.west);

				\node[roundnode]  (E1) at (4.8, -3.7) {};
				\node[roundnode]   (F1) at (4.8, -5.3) {};
				\node[roundnode]   (G1) at (6.1, -3.7) {};
				\node[roundnode]   (H1) at (6.1, -5.3) {};

				\path (E1.west) edge [out=220,in=140] node [midway, left] {$a_3$} (F1.west);
				\path[blue] (G1.east) edge [out=320,in=40] node [midway, right] {$b_1$} (H1.east);
				\draw[thick, blue] (E1) -- (F1) node [midway, right=2pt] {$c_2$};
				\draw[thick] (G1) -- (H1) node [midway, left=2pt] {$d_2$};

				\node[roundnode]  (E2) at (8.3, -3.7) {};
				\node[roundnode]   (F2) at (8.3, -5.3) {};
				\node[roundnode]   (G2) at (9.6, -3.7) {};
				\node[roundnode]   (H2) at (9.6, -5.3) {};

				\path[blue] (E2.west) edge [out=220,in=140] node [midway, left] {$a_2$} (F2.west);
				\draw[thick] (E2) -- (F2) node [midway, right=2pt] {$c_1$};
				\draw[thick] (G2) -- (H2) node [midway, left=2pt] {$d_1$};

				\node[roundnode]  (E3) at (11, -3.7) {};
				\node[roundnode]   (F3) at (11, -5.3) {};
				\path[blue] (E3.west) edge [out=220,in=140] node [midway, left] {$a_1$} (F3.west);
			\end{tikzpicture}}
\caption{\label{fig:matroidUnfoldingFigure} An example of matroid unfolding for two graphic matroids $\matroid'_1$, $\matroid'_2$. In a graphic matroid, the elements are given by the edges of a graph and the independent sets are the forests. The labels of edges in $\matroid'_1,\matroid'_2$ contain the name of the edge as well as the associated weight. The blue edges show a maximum weight common indpendent set in $\matroid'_1,\matroid'_2$ and a maximum cardinality common independent set with the same value in the unfolded matroids $\matroid_1$, $\matroid_2$. The cycles of the new graphic matroids correspond to cycles in the original instance.} 
		\end{figure}
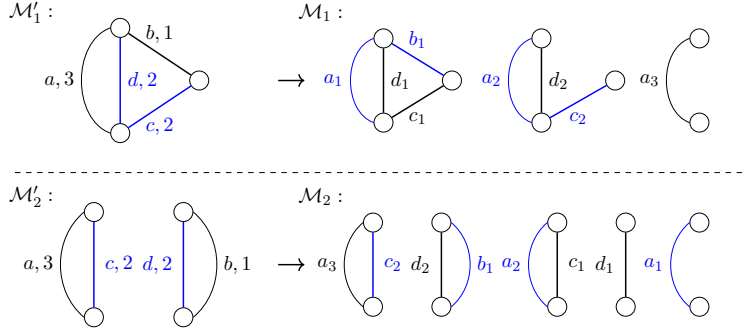

\paragraph{Refolding. } Suppose $\matroid_1,\matroid_2$ are the unfolded matroids of $\matroid'_1,\matroid'_2$. In the final step of the reduction, we are able to show that given $I\in \independentsets_1\cap\independentsets_2$, we can refold it back in $\matroid'_1,\matroid'_2$. Moreover, we show that the refolded version of $I$ denoted $\groundset'_I$ is such that a) $|\groundset'_I|\leq W\cdot |I|$ b) and there exists $I'\subseteq \groundset_i$ such that $I'\in \independentsets'_1\cap\independentsets'_2$ such that $w(I')=|I|$. To extract such a set from $\groundset'_I$, we give an efficient algorithm (\Cref{lem:weighted_additive}) based on the techniques of \cite{BlikstadT25}.

\subsection{Preliminaries}

We now give some remaining definitions (see e.g. \cite{Schrijver03}).

\begin{definition}[Span]
 We say that a set $S \subseteq \mathcal{N}$ \emph{spans} a set $T \subseteq \mathcal{N}$ if for all $e \in T$ we have $S \cup \{e\} \not \in \mathcal{I}$.
\end{definition}

\begin{definition}[Circuits]
A set $S \subseteq \mathcal{N}$ is a \emph{circuit} if and only if $S \not \in \mathcal{I}$ but for any $T\subsetneq S$, $T\in \mathcal{I}$.
\end{definition}

\pagebreak

\begin{property2}[Properties of Circuits]\label{prop:prop_circuit}
A matroid can also be characterized by its circuits. The set of all circuits of $\mathcal{M}$, denoted $C(\mathcal{M})$ satisfy the following properties.
\begin{enumerate}[label=(\alph*)]
\item For $C_1,\mathcal{C}_1\in C(\mathcal{M})$ such that $C_1\subseteq C_2$, then $C_1=C_2$. \label{prop:prop_circuit:item:prop_circuit1}
\item If there are circuits $C_1, C_2 \in C(\mathcal{M})$ with $C_1 \neq C_2$ and $x \in C_1 \cap C_2$, then for all $y\in C_1\setminus C_2$, there exist $C_y\in C(\matroid)$ such that $C_y\subseteq C_1\cup C_2\setminus \{x\}$ and $y \in C_y$.\label{prop:prop_circuit:item:prop_circuit2}
\end{enumerate}
\end{property2}

\paragraph{Roadmap.} In the subsequent section we prove various subroutines used in our reduction. In \Cref{sec:matroidunfolding}, we will discuss matroid unfolding. In \Cref{,sec:aspectratio} we prove theorems about our aspect-ratio reduction. In \Cref{sec:refolding} we discuss our refolding procedure In \Cref{,sec:mwmadditive} we give an efficient algorithm for computing a $(1-\varepsilon)$-approximation to the maximum weight independent set for the case when the groundset is small. In \Cref{sec:applications}, we will show some implications of our results in static, streaming and one-way communication complexity models. Our applications are summarized in \Cref{table:summary of results}.

\section{Ingredients for Our Framework}

In this section, we proceed by showing structural results related to our framework. We first start with matroid unfolding.

\subsection{Matroid Intersection Unfolding}
\label{sec:matroidunfolding}
Suppose we are given an instance of integer-weighted matroid intersection $\matroid'_1=(\groundset',\independentsets'_1,w)$ and $\matroid'_2=(\groundset',\independentsets'_2,w)$ such that $w:\groundset'\rightarrow\set{1,2,\cdots, W}$. We will obtain an unweighted instance $\matroid_1=(\groundset,\independentsets_1)$ and $\matroid_2=(\groundset,\independentsets_2)$, such that $\card{\groundset}\leq W\cdot \card{\groundset'}$.

The idea is based on graph unfolding~\cite{KaoLST01} and will let us calculate the value of the maximum weight common independent set $S^*$. Furthermore, it allows us to identify the most relevant elements of the matroid containing $S^*$. We begin by explaining how to unfold a weighted matroid intersection instance. Suppose we are given two matroids $\matroid'_1=(\groundset',\independentsets'_1,w)$ and $\matroid'_2=(\groundset',\independentsets'_2,w)$ with integer weights with maximum weight common independent set of weight $\opt'$. We create matroids $\matroid_1=(\groundset,\independentsets_1)$ and $\matroid_2=(\groundset,\independentsets_2)$ by replacing each element $e\in \groundset'$ with $w(e)$ copies of the element $\{e_{1},\cdots, e_{w(e)}\}$. However, we need to impose additional structure on $\matroid_1,\matroid_2$ so that \begin{inparaenum}[(a)]
    \item\label{item:optabs} $\opt=\opt'$ \item\label{item:queriesabs} We can implement independence (resp. rank) queries to $\matroid_1,\matroid_2$ using a small number of independence (resp. rank) queries to $\matroid'_1,\matroid'_2$.
\end{inparaenum} Formally, the new matroids are defined as follows.

\begin{definition}[Unfolded Matroid Intersection]\label{def:matroidunfolding}
    Suppose we have an instance $\matroid'_{1}=(\mathcal{N}',\mathcal{I}'_1,w)$ and $\matroid'_2=(\groundset',\independentsets'_2,w)$ of integer-weighted matroid intersection, such that for all $e\in \mathcal{N}'$, $w(e)\in \{1,2,\cdots, W\}$. We create an instance $\matroid_1=(\groundset,\independentsets_1)$ and $\matroid_2=(\groundset,\independentsets_2)$ of unweighted matroid intersection as follows.
    \begin{enumerate}[label=(\alph*)]
        \item\label{def:matroidunfolding:item:groundset} For each $e\in \groundset'$, add $w(e)$ many corresponding elements, $e_1,\cdots, e_{w(e)}$ to $\groundset$. %Let $\groundset_i=\set{e_i\mid e\in \groundset', w(e)\geq i}$.
        \item\label{def:matroidunfolding:item:matroid1} Consider $I\subseteq \groundset $. For all $i\in [W]$, let $I'_{1,i}=\set{e\mid e_i\in I}$. We say $I\in \independentsets_1$ if and only if $I'_{2,i}\in \independentsets'_1$ for all $i\in [W]$.
        \item\label{def:matroidunfolding:item:matroid2} Consider $I\subseteq \groundset$. For all $i\in [W]$, let $I'_{2,i}=\set{e\mid e_{w(e)-i+1}\in I}$. We say $I\in \independentsets_2$ if and only if $I'_{2,i}\in \independentsets'_2$ for all $i\in [W]$.
    \end{enumerate}
\end{definition}

Having defined the unfolded instance, we now show that they are valid matroids.  

\begin{lemma}\label{lem:validityofreduction}
Suppose we have an instance $\matroid'_1=(\groundset',\independentsets'_1,w)$, $\matroid'_2=(\groundset',\independentsets'_2,w)$ of weighted matroid intersection. Suppose that the weights $w(e)\in \set{1,2,\cdots, W}$ for all $e\in \groundset'$. Let $\matroid_1=(\groundset,\independentsets_1)$ and $\matroid_2=(\groundset,\independentsets_2)$ be the unweighted instance obtained by applying \Cref{def:matroidunfolding}. Then, $\matroid_1$ and $\matroid_2$ are matroids.
\end{lemma}
\begin{proof}
    In order to prove the claim, it is sufficient to verify that $\matroid_1$ and $\matroid_2$ satisfy the three matroid properties. We start by showing this for $\matroid_1$.
    \begin{enumerate}[label=(\alph*)]
        \item{\bf Empty Set Property:} Note $\phi\subset \groundset$ satisfies the conditions imposed by \Cref{def:matroidunfolding}\ref{def:matroidunfolding:item:matroid1}. Thus, $\phi \in \independentsets_1$. 
        \item{\bf Downward Closure:} We show that the downward closure property holds for $\matroid_1$. Suppose $A\subseteq B\subseteq \groundset$ and $B\in \independentsets_1$. Then, we would like to show that $A\in \independentsets_1$ as well. To show this, let $A'_i$ and $B'_i$ as in \Cref{def:matroidunfolding}\ref{def:matroidunfolding:item:matroid1}. Note that since $B\in \mathcal{I}_1$, this implies that $B'_i\in \independentsets_1'$. By downward closure of $\matroid'_1$, $A'_i\in \independentsets_1'$ as well. Thus, we can conclude that $A\in \independentsets_1$. 
        \item {\bf Exchange Property:} Suppose $A,B\in \independentsets_1$ such that $\card{A}>\card{B}$. We define $A_i',B_i'\in \independentsets_1'$ for $i\in [W]$ similarly as in \Cref{def:matroidunfolding}\ref{def:matroidunfolding:item:matroid1}. There exists $i\in [W]$ such that $|A_i'|>|B_i'|$. Since $\matroid_1'$ is a matroid and $A_i'$, $B_i'$ are independent, they therefore satisfy the exchange property. This implies there exists $x\in A_i'\setminus B_i'$ such that $B_i'\cup \{x\}\in \independentsets_i'$. It must be the case that $x_i\in A\setminus B$ and $B\cup\{x\}\in \independentsets_1$.
    \end{enumerate}
We now verify these properties for $\matroid_2$. 
\begin{enumerate}[label=(\alph*)]
    \item {\bf Empty Set Property:} The set $\phi\subset \groundset$ satisfies the conditions imposed by \Cref{def:matroidunfolding}\ref{def:matroidunfolding:item:matroid2} and therefore, $\phi\in \independentsets_2$. 
    \item {\bf Downward Closure:} We show that the downward closure property holds for $\matroid_2$. Suppose $A\subseteq B\subseteq \groundset$ and $B\in \independentsets_2$. For all $i\in [W]$, let $A_i',B_i'\subseteq \groundset'$ as in \Cref{def:matroidunfolding}\ref{def:matroidunfolding:item:matroid2}. Since $B\in \independentsets_2$, $B_i'\in \independentsets_2'$ for all $i\in [W]$ as well. Since $A_i'\subseteq B_i'$ and $\matroid_2'$ is a matroid, we have $A_i'\in \independentsets'_2$. Consequently, $A\in \independentsets_2$. 
    \item {\bf Exchange Property:} Suppose $A,B\in \independentsets_2$ such that $|A|>|B|$. We define $A_i', B_i
'\in \independentsets_2'$ for $i\in [W]$ as in \Cref{def:matroidunfolding}\ref{def:matroidunfolding:item:matroid2}. There exist $i\in [W]$ such that $|A_i'|>|B_i'|$. Since $\matroid_2'$ is a matroid, and $A_i', B_i'\in \independentsets_2'$, this implies that there exists $x\in A_i'\setminus B_i'$ such that $B_i'\cup \{x\}\in \independentsets'_2$. It must be the case that $x_i\in A\setminus B$ and $B\cup\{x\}\in \independentsets_2$.  
\end{enumerate}
Hence, the result of unfolding is a valid unweighted matroid intersection instance.
\end{proof}

\subsubsection{Oracle Access to Unfolded Matroids.}

To check whether a given set is independent in the resulting matroids, it suffices to see if $O(W)$ subsets, each containing at most one copy of each element, are independent in the original matroid. Hence, we immediately have an independence oracle for the new matroid.

\begin{claim2}[Independence Queries]\label{lem:oraclecallstoimplementunfolded}
Suppose $\matroid'_1=(\groundset',\independentsets'_1,w)$ and $\matroid'_2=(\groundset',\independentsets'_2,w)$ are an instance of integer-weighted matroid intersection with weights in $\set{1,\cdots, W}$. Let $\matroid_1=(\groundset,\independentsets_1)$ and $\matroid_2=(\groundset,\independentsets_2)$ be obtained by applying the unfolding procedure in \Cref{def:matroidunfolding}\ref{def:matroidunfolding:item:matroid1} and \Cref{def:matroidunfolding}\ref{def:matroidunfolding:item:matroid2}, respectively. Then, we can implement an independence query $I\in \independentsets_1$ via at most $W$ calls to the independence oracle for $\matroid'_1$. Similarly, we can implement an independence query $I\in \independentsets_2$ via at most $W$ calls to the independence oracle for $\matroid'_2$.
\end{claim2}

Let $\matroid'_1=(\groundset',\independentsets'_1,w)$ and $\matroid'_2=(\groundset',\independentsets'_2,w)$ be weighted matroids. Let $\matroid_1=(\groundset,\independentsets_1)$ and $\matroid_2=(\groundset,\independentsets_2)$ be the unfolded versions. We can implement the rank oracle for $\matroid_1,\matroid_2$ as follows. Let $S\subseteq\groundset$ and for $i\in [W]$, consider $S'_i\subseteq \groundset$ as per definition \Cref{def:matroidunfolding}\ref{def:matroidunfolding:item:matroid1}. By definition of $\matroid_1$, we have, $\rk_1(S)=\sum_{i\in [W]}\rk_1'(S_j)$. A similar observation holds for $\matroid_2$ as well.
\begin{claim2}[Rank Queries]\label{lem:oraclecallsRanktoimplementunfolded}
Let $\matroid'_1=(\groundset',\independentsets'_1,w)$ and $\matroid'_2=(\groundset',\independentsets'_2,w)$ be an instance of integer-weighted matroid intersection with weights in $\set{1,2,\cdots, W}$. Let $\matroid_1=(\groundset,\independentsets_1)$ and $\matroid_2=(\groundset,\independentsets_2)$ be obtained by applying the unfolding procedure in \Cref{def:matroidunfolding}\ref{def:matroidunfolding:item:matroid1} and \Cref{def:matroidunfolding}\ref{def:matroidunfolding:item:matroid2}, respectively. Then, we can implement a rank query for $I\in \independentsets_1$ via at most $W$ calls to the rank oracle for $\matroid'_1$. Similarly, we can implement a rank query for $I\in \independentsets_2$ via at most $W$ calls to the rank oracle for $\matroid'_2$.
\end{claim2}

\subsubsection{Equivalence of Original and Unfolded Matroids.}

We would now like to show that the optimal solution to the unfolded matroid intersection problem has the same value as the weight of the optimal solution for the original weighted matroid intersection problem. 

\begin{theorem}\label{thm:equivalence}
Let $\matroid_1'=(\groundset', \independentsets_1',w)$ and $\matroid_2'=(\groundset', \independentsets_2',w)$ be an instance of integer-weighted matroid intersection with optimal solution of value $\opt'$. Let $\matroid_1=(\groundset, \independentsets_1)$ and $\matroid_2=(\groundset, \independentsets_2)$ be the unweighted matroid intersection instance obtained by unfolding $\matroid_1'$ and $\matroid_2'$. Suppose $\opt$ denote the optimal solution of the unweighted instance. Then $\opt=\opt'$. 
\end{theorem}

We will first show that the value of optimal solution in the unfolded instance provides an upper bound to the value of optimal solution in the original weighted instance, as described in the following lemma.

\begin{lemma}\label{lem:optgeqopt'}
    Let $\matroid_1'=(\groundset', \independentsets_1',w)$ and $\matroid_2'=(\groundset', \independentsets_2',w)$ be an instance of integer-weighted matroid intersection with optimal solution of value $\opt'$. Let $\matroid_1=(\groundset, \independentsets_1)$ and $\matroid_2=(\groundset, \independentsets_2)$ be the unweighted matroid intersection instance obtained by unfolding $\matroid_1'$ and $\matroid_2'$. Suppose $\opt$ denote the optimal solution of the unweighted instance. Then $\opt\geq\opt'$. 
\end{lemma}

\begin{proof}
    Now, we want to argue that the value of the common independent set of largest weight for weighted matroids $\mathcal{M}'_1,\mathcal{M}'_2$ is equal to the size of the largest common independent set of $\mathcal{M}_1,\mathcal{M}_2$. First, we argue that the largest common independent set in $\mathcal{M}_1,\mathcal{M}_2$ has at least $\opt'$ many elements. This will show that $\opt\geq \opt'$. Let $S'\in \independentsets'_1\cap \independentsets_2'$ be such that $w(S')=\opt'$. Let $S\subseteq \groundset$ be defined as follows: $S=\cup_{e\in S'}\cup_{i\in [w(e)]}\{e_i\}$. Note that $\card{S}=w(S')$. We now only need to argue that $S\in \independentsets_1\cap \independentsets_2$. To show $S\in \independentsets_1$, we define for all $i\in [W]$, $S_i'\subseteq \groundset'$ as in \Cref{def:matroidunfolding}\ref{def:matroidunfolding:item:matroid1}. Since $S_i'\subseteq S'$, we can conclude by downward closure property of matroids that $S_i'\in \independentsets'_1$ for all $i\in [W]$. Thus, we can conclude that $S\in \independentsets_1$. The argument for showing $S\in \independentsets_2$ follows analogously.
\end{proof}
In order to prove the converse, we will take an optimal dual solution for the weighted problem $y',z'$ and input them to \Cref{alg:dualsolution}.  The algorithm outputs $y,z$, which we will show are feasible solutions to the unfolded matroid intersection problem (\Cref{lem:feasibilityyz}). Additionally, $f(y',z')\geq g(y,z)$ (\Cref{lem:valueofoutputduals}). Then, by weak duality, we can conclude that $\opt\leq \opt'$. We start by giving some intuition about \Cref{alg:dualsolution}. The algorithm ``splits'' the value of $y'(S')$ (and $z'(S')$) for each $S' \subseteq \groundset'$ among a collection of sets $\mathcal{S}_{y}$ (and $\mathcal{S}_{z}$), such that $\sum_{S \in \mathcal{S'}_{y}} y(S) \leq y'(S')$ (and analogous for $z'$ and $\matroid_2$). For a set $S \in \textup{supp}(y')$ the sets in these collections are the \emph{relevant sets for $S$} (the exact definition is below). A single iteration of the for-loop in line \cref{line:forloopy} implements this process for $y'$ and a single set in the $\textup{supp}(y')$, while an iteration of the for-loop in line \ref{line:forloopz} is responsible for a single set in the $\textup{supp}(z')$. We process the sets in the order provided by the chain, starting with the largest set. Hence, the $i$-th iteration of the for-loop in line \ref{line:forloopy} assigns a $y$-value of 1 to exactly $y'(S^i)$ subsets of $\groundset$, containing only copies of elements in $S^i$ with the same index. Similarly, the $i$-th iteration of the for-loop in line \ref{line:forloopz} assigns a $z$-value of 1 to exactly $z'(T^i)$ subsets of $\groundset$, the relevant sets for $T^i$, containing only copies of elements in $T^i$ with the same difference between the original weight and their index. We start with giving some observations about \Cref{alg:dualsolution}.

\begin{observation}[Relevant Sets for $S^i$]\label{obs:descriptionofinheritedsetsy}
    For $i\in [l]$, and $f_y\in \set{\sum_{j=1}^{i-1}y'(S^j)+1,\cdots, \sum_{j=1}^i y'(S^j)}$, $S^i_{f_y}=\{e_{f_y}\mid e\in S^i, w(e)\geq f_y\}$. 
\end{observation}

\begin{observation}[Relevant Sets for $T^i$]\label{obs:descriptionofinheritedsetsz}
    For $i\in [l]$, and \\$f_z\in \set{\sum_{j=1}^{i-1}z'(T^j)+1,\cdots, \sum_{j=1}^i z'(T^j)}$,$T^i_{f_z}=\{e_{w(e)+1-f_z}\mid e\in T^i, w(e)\geq f_z\}$. 
\end{observation}

\begin{observation}\label{obs:coveredelementsy}
    Suppose we are in the $i$th execution of the for loop in \Cref{alg:dualsolution} Line \ref{line:forloopy}. At the end of the execution, $f_y=\sum_{j=1}^iy'(S^j)$ and for each $e\in \groundset'$, elements $\set{e_1,\cdots,e_{f_y}}\subseteq \groundset$ are covered by $y$ and therefore, $U^y\cap \set{e_1,\cdots,e_{f_y}}=\emptyset$. This follows from the fact that the sequence $S^1,S^2,\cdots, S^l$ forms a chain. 
\end{observation}

\begin{observation}\label{obs:coveredelementsz}
    Suppose we are in the $i$th execution of the for loop in \Cref{alg:dualsolution} Line \ref{line:forloopz}. At the end of the execution, $f_z=\sum_{j=1}^iz'(T^j)$ and for each $e\in \groundset'$, elements $\set{e_{w(e)},\cdots,e_{w(e)+1-f_z}}\subseteq \groundset$ are covered by $z$. Therefore, $U^z\cap \set{e_{w(e)},\cdots,e_{w(e)+1-f_z}}=\emptyset$. This follows from the fact that the sequence $T^1,T^2,\cdots, T^l$ forms a chain. 
\end{observation}

We are now ready to show the feasibility property of the output $y,z$ obtained by running \Cref{alg:dualsolution} on an integral optimal dual solution $y',z'$, such that \textup{supp}($y'$) and \textup{supp}($z'$) are chains. \Cref{lem:chaineddualsexistence} ensures that such $y',z'$ always exist.

\begin{lemma}[Feasibility of Output Duals]\label{lem:feasibilityyz}
    Suppose $\matroid'_1=(\groundset',\independentsets'_1,w)$ and $\matroid'_2=(\groundset',\independentsets'_2,w)$ is an instance of integer-weighted matroid intersection problem. Let $y',z'$ be optimal dual solutions, such that \textup{supp}($y'$) and \textup{supp}($z'$) are chains, input to \Cref{alg:dualsolution}. Then the output $y,z$ is a feasible dual for the $\matroid_1,\matroid_2$, the unfolded matroids.
\end{lemma} 
\begin{proof}
    To show this claim, it is sufficient to show that for all $e\in \groundset'$, for all $i\in[w(e)]$, $e_i\notin U^z\cap U^y$. Assume to the contrary that there exist such an $e_i$. By Observation~\ref{obs:coveredelementsy} and Observation~\ref{obs:coveredelementsz}, we have,
    \begin{align*}
        i \geq \sum_{j=1}^ly'(S^j)+1, \text{ and, }i \leq  w(e) -\sum_{j=1}^k z'(T^j) 
    \end{align*}
Combining this, we have,
\begin{align*}
    \sum_{j=1}^l y'(S^j)+\sum_{j=1}^k z'(T^j)\leq w(e)-1
\end{align*}
Since $\{S^j\}_{j\in [l]}$ and $\{T^j\}_{j\in [k]}$ form the supp($y'$) and supp($z'$), respectively, we can conclude that,
\begin{align*}
    \sum_{S:e\in S}y'(S)+z'(S)\leq w(e)-1,
\end{align*}
contradicting the fact that $y',z'$ are feasible solutions for $\matroid'_1$ and $\matroid'_2$.
\end{proof}

\begin{algorithm}
\algorithmicrequire{ Integral optimal dual solutions $y',z'$ for $\matroid'_1,\matroid'_2$,\\  \hspace*{4.4em} \textup{supp}($y'$), \textup{supp}($z'$) form a chain}\\
	\algorithmicensure{ Dual solutions $y,z$ for unfolded matroids $\matroid_1,\matroid_2$)}
	\caption{\textsc{UnweightedDual}()}
	\begin{algorithmic}[1]
    \State Let $S^{l}\subsetneq S^{l-1}\subsetneq \cdots \subsetneq S^1$ be the sets in $\text{supp}(y')$.
    \State Let $T^{k}\subsetneq T^{k-1}\subsetneq \cdots \subsetneq T^1$ be the sets in $\text{supp}(z')$.
    \State Initialize $y\leftarrow 0$, $z\leftarrow 0$.
    \State Initialize $U^y,U^z\leftarrow \groundset$. \Comment{These are elements uncovered by duals.}
    \State Initialize $f_y,f_z\leftarrow 1$.
    \For{$i=1,2,\cdots, l$}\label{line:forloopy}
    \While{$f_y\leq \sum_{j=1}^i y'(S^j)$}\label{line:whileloop}
    \State Initialize $S^i_{f_y}\leftarrow \emptyset$. \Comment{This will be the set of elements which will be covered by $y$.}
    \For{$e\in S^i$ with $w(e)\geq f_y$} 
    \State $S^i_{f_y}\leftarrow S^i_{f_y}\cup\{e_{f_y}\}$.
    \EndFor
    \State $y(S_{f_y}^i)\leftarrow 1$. \label{line:assigndualy}
    \State $U^y\leftarrow U^y\setminus S^i_{f_y}$.
    \State $f_y\leftarrow f_y+1$.
    \EndWhile\label{line:iterationIcovered}
    \State $i\leftarrow i+1$.
    \EndFor
    \For{$i=1,2,\cdots, k$}\label{line:forloopz}
    \While{$f_z\leq \sum_{j=1}^i z'(T^j)$}\label{line:whileloop2}
    \State Initialize $T^i_{f_z}\leftarrow \emptyset$. \Comment{This will be the set of elements which will be covered by $z$.}
    \For{$e\in T^i$ with $w(e)\geq f_z$}
    \State $T^i_{f_z}\leftarrow T^i_{f_z}\cup\{e_{w(e)+1-f_z}\}$.
    \EndFor
    \State $z(T_{f_z}^i)\leftarrow 1$.\label{line:assigndualz}
    \State $U^z\leftarrow U^z\setminus T^i_{f_z}$.
    \State $f_z\leftarrow f_z+1$.
    \EndWhile \label{line:iterationIcoveredbyz}
    \State $i\leftarrow i+1$.
    \EndFor
    \State Return $y,z$.
	\end{algorithmic}
     \caption{Algorithm to Obtain Feasible Dual Solution for Unfolded Matroids $\matroid_1,\matroid_2$}
      \label{alg:dualsolution}
\end{algorithm}

We are now left with the task of showing the following lemma.

\begin{lemma}\label{lem:valueofoutputduals}
    Let $\matroid'_1=(\groundset',\independentsets_1',w)$ and $\matroid'_2=(\groundset',\independentsets_2',w)$ be an instance of matroid intersection. Let $\rk_1',\rk'_2$ be the rank functions associated with them. Let $\matroid_1=(\groundset,\independentsets_1)$ and $\matroid_2=(\groundset,\independentsets_2)$ with associated rank functions $\rk_1,\rk_2$ be the matroids obtained via the unfolding process. Let $y',z'$ be optimal dual solutions for $\matroid'_1,\matroid'_2$, such that \textup{supp}($y'$) and  \textup{supp}($z'$) each form a chain. Let $y,z$ be feasible dual solutions obtained as an output by inputing to \Cref{alg:dualsolution} duals $y',z'$, such that \textup{supp}($y'$) and \textup{supp}($z'$) each form a chain. Then,
    \begin{align*}
        \sum_{S\subseteq\groundset} y(S)\rk_1(S)+z(S)\rk_2(S)\leq \sum_{S'\subseteq\groundset' }y'(S')\rk'_1(S')+z'(S')\rk'_2(S').
    \end{align*}
\end{lemma}
\begin{proof}
    \sloppy Consider a set $S^i$ in supp($y'$) which is assigned dual values $y'(S^i)$. The sets $S^i_{f_y}$ for $f_y\in  \set{\sum_{j=1}^{i-1}y'(S^j)+1,\cdots, \sum_{j=1}^i y'(S^j)}$ together inherit the value from $y'(S_i)$, since $y(S^i_{f_y})=1$ for each such set. By Observation~\ref{obs:descriptionofinheritedsetsy} and \Cref{def:matroidunfolding}\ref{def:matroidunfolding:item:matroid1}, we can deduce that $\rk'_1(S^i)\geq \rk_1(S^i_{f_y})$ for all $f_y\in \set{\sum_{j=1}^{i-1}y'(S^j)+1,\cdots, \sum_{j=1}^i y'(S^j)}$. Thus, we can conclude that,
    \begin{align*}
        \sum_{S'\subseteq \groundset'} y'(S')\rk'_{1}(S')&=\sum_{j=1}^l y'(S^i)\rk'_1(S^i )\\
        &\text{(Since $\set{S^1,S^2,\cdots,S^l}=\text{supp}(y')$)}\\
        &\geq \sum_{j=1}^l \sum_{f_y=\sum_{j=1}^{i-1}y'(S^j)+1}^{\sum_{j=1}^i y'(S^j)} \rk_{1}(S^i_{f_y})\\
        &\text{(Since $\rk'_1(S^i)\geq \rk_{1}(S^i_{f_y})$)}\\
        &= \sum_{j=1}^l \sum_{f_y=\sum_{j=1}^{i-1}y'(S^j)+1}^{\sum_{j=1}^i y'(S^j)} y(S^i_{f_y})\rk_{1}(S^i_{f_y})\\
        &\text{(Since $y(S^i_{f_y})=1$ by \Cref{line:assigndualy})}\\
        &=\sum_{S\subseteq \groundset} y(S)\rk_1(S).
    \end{align*}
By a similar argument, we can conclude that,
\begin{align*}
    \sum_{S'\subseteq \groundset'} z'(S')\rk'_{2}(S')\geq\sum_{S\subseteq\groundset }z(S)\rk_2(S).
\end{align*}
Combining the two equations, we have our argument.
\end{proof}

We now show \Cref{thm:equivalence}.

\begin{proof}[Proof of \Cref{thm:equivalence}]
    Note that $\opt\geq\opt'$ follows from \Cref{lem:optgeqopt'}. To show the converse, let $y',z'$ be an optimal integral dual solution for $\matroid_1',\matroid_2'$, such that \textup{supp}($y'$) and \textup{supp}($z'$) each form a chain. Observe that we can construct a feasible dual $y,z$ solution for $\matroid_1,\matroid_2$ such that $\sum_{S\subseteq\groundset} y(S)\rk_1(S)+z(S)\rk_2(S)\leq \sum_{S'\subseteq\groundset' }y'(S')\rk'_1(S')+z'(S')\rk'_2(S')$. Thus, we have,
    \begin{align*}
        \opt&\leq \sum_{S\subseteq\groundset} y(S)\rk_1(S)+z(S)\rk_2(S)\\
        &\text{(By weak duality)}\\
    &\leq\sum_{S'\subseteq\groundset' }y'(S')\rk'_1(S')+z'(S')\rk'_2(S')\\
    &\text{(By \Cref{lem:valueofoutputduals})}\\
    &=\opt'\\
    &\text{(As $y',z'$ is optimal by strong duality)}
    \end{align*}
This concludes our proof.
\end{proof}

Finally, we note that after unfolding both the size of the groundset and the size of the set giving the optimal solution to the respective relevant problem increase as follows.
\begin{observation}[Unfolded Matroids Size]\label{obs:unfoldedmatroidsize}
    Let $\matroid'_1=(\groundset',\independentsets'_1,w),\matroid'_2=(\groundset',\independentsets'_2,w)$ be integer-weighted matroids. Let $r$ be the rank of the intersection of $\matroid'_1,\matroid'_2$. Suppose $\matroid_1=(\groundset,\independentsets_1)$ and $\matroid_2=(\groundset,\independentsets_2)$. Then, $\card{\groundset}\leq W\cdot\card{\groundset'}$ and $\opt\leq W\cdot r$.
\end{observation}

\subsection{Aspect Ratio Reduction}
\label{sec:aspectratio}

In this section, we discuss how to reduce the aspect ratio for a special class of matroids called $\varepsilon^{-1}$-spread matroids. Subsequently, we will show how to convert arbitrary weighted matroids into $\varepsilon^{-1}$-spread matroids. 

\subsubsection{Reducing the Aspect Ratio for $\varepsilon^{-1}$-spread Matroids.}
\label{subsec:spreadmatroids}
We begin by defining the notion of $\varepsilon^{-1}$-spread matroids. This is a useful notion because it is relatively easier to find weighted common independent sets in such matroids. 

\begin{definition}[$\varepsilon^{-1}$-spread matroid]
Let $[l_1,r_1],[l_2,r_2],\cdots, [l_k,r_k]$ be disjoint weight intervals with the property that $l_{i+1}>r_{i}\varepsilon^{-1}$. Let $\matroid'_1=(\groundset',\independentsets'_1,w)$ and $\matroid'_2=(\groundset',\independentsets'_2,w)$ be a matroid intersection instance such that $\groundset'_i=\{e \in \groundset' \mid w(e) \in [l_i,r_i]\}$ form a partition of $\groundset'$. Then, $\matroid'_1,\matroid'_2$ are called $\varepsilon^{-1}$-spread matroids.
\end{definition}

Let $\matroid'_1,\matroid'_2$ be an instance of $\varepsilon^{-1}$-spread matroid intersection. For $j\in [k]$ let $L_{j}=[l_j,r_j]$ be the corresponding intervals and define $\matroid'|L_j := \matroid | \{e \in \groundset' \mid w(e) \in L_j\}$ for a matroid $\matroid' = (\groundset', \independentsets', w)$ and a weight interval $L_j$. In addition, let $I'_{j}$ be the optimal common independent set of $\matroid'_1|L_j$ and $\matroid_{2}'|L_j$. Suppose we proceed greedily to combine $\{I'_{j}\}_{j\geq 0}$, in descending order of intervals and only proceeding to add an element $e$ to the combined set $I_{\text{merge}}$ if and only if $I_{\text{merge}}\cup \{e\}$ is a common independent set of $\matroid'_1,\matroid'_2$. Now we want to bound the loss we incur due to the combining step as given in the following lemma.

\begin{lemma}\label{lem:greedymerge}
Let $\matroid_1',\matroid_2'$ be an instance of weighted matroid intersection that is $\varepsilon^{-1}$-spread. For $j\in [k]$, let $L_j=[l_j,r_j]$ be the corresponding intervals and $I_j'$ be a common independent set of $\matroid_1'|L_j$ and $\matroid_2'|L_j$.
Suppose $I_{\text{merge}}$ is obtained by greedily combining $I'_j$ for $j\in [k]$. Then,
\begin{align*}
    w(I_{\text{merge}})\geq (1-4\varepsilon)\sum_{j=1}^k w(I'_j).
\end{align*}
\end{lemma}

In order to show this, we will prove that we can charge the weight of an element in $I_j'$, which was not included, to an element in $I_\text{merge} \cap I_i'$ for some $i > j$. Furthermore, there is an assignment such that each element in $I_\text{merge}$ is charged by at most two elements from each set $I_j'$ for $j \geq 0.$ In the argument we will use the following structural lemma.

\begin{lemma}\label{lem_charging}
    Let $\matroid'=(\groundset',\independentsets',w)$ be a $\varepsilon^{-1}$-spread matroid with corresponding weight intervals $L_{i}=[l_i,r_i]$ for $i\geq 0$. Let $I'$ be an independent set in the matroid $\mathcal{M}'$  such that $I' \subseteq \bigcup_{t \geq j} S'_t$, where $S'_t$ is an independent set in $\mathcal{M}' | L_t$. Let further $\{e_1, \cdots e_l\} \subseteq S'_j$ be the set of elements in the independent set of level $j$ that are spanned by $I'$. Let $C_t$ be the unique circuit in $I' \cup \{e_t\}$ for every $t \in [l]$, then $\card{\left(\bigcup_{t=1}^l C_t \right) \cap (I' \setminus S'_j)} \geq l$.
\end{lemma}

\begin{proof}
We proceed by contradiction to show that if $\card{\left(\bigcup_{t=1}^l C_t \right) \cap (I' \setminus S'_j)} < l$ then, there must be a circuit in $S'_j$, thus contradicting the fact that $S'_j$ is an independent set. To show this, we proceed by induction. 

First let us show the base case. Consider any two circuits and without loss of generality, let these be $C_1$ and $C_2$. Observe that for both $C_1$ and $C_2$, we have, $\card{C_1\cap I'\setminus S'_j}\geq1$ and $\card{C_2\cap I'\setminus S'_j}\geq 1$. If this is not the case, then this would imply either $C_1\subseteq S'_j$ or $C_2\subseteq S'_j$ and we would have a contradiction. Assume that $\card{(C_1\cup C_2)\cap I'\setminus S'_j}=1$. Let this element be $a$. Thus, we have, $C_1\cup C_2\setminus \{a\}\subseteq S'_j$. Since $C_1\neq C_2$ by \Cref{prop:prop_circuit}\ref{prop:prop_circuit:item:prop_circuit2}, we can conclude that $C_1\cup C_2\setminus \{a\}\subseteq S'_j$ contains a circuit. 

By inductive hypothesis, we can assume that the claim applies to any $t<l$ circuits from $\{C_1,\cdots, C_l\}$. We will now show the claim for any subset of $t+1$ circuits. Without loss of generality, let these $t+1$ circuits be $C_1,\cdots, C_{t+1}$. We additionally know that for all $s\in [t+1]$, $\card{C_s\cap I'\setminus S'_j}\geq 1$. Thus, if $\card{\cup_{i=1}^{t+1}C_{i}\cap I'\setminus S'_j}<t+1$, this implies that there exist an element $b\in I'\setminus S'_j$ such that at $b\in C_j$ for at least two $C_j$. We define $\mathcal{C}_b=\set{C_i\mid b\in C_i, i\in [t+1]}$ and without loss of generality, assume $C_1\in \mathcal{C}_b$. For $C_t\in \mathcal{C}_b\setminus \{C_1\}$, let $C_t'\subseteq (C_t\cup C_1)\setminus \{b\}$ be a circuit containing $e_t$ (this exists by \Cref{prop:prop_circuit}\ref{prop:prop_circuit:item:prop_circuit2}). Let this collection be called $\mathcal{C}'_b$. Observe that $C_s'\neq C_{t}'$ as $e_s\in C_s'$ and $e_t\in C_t'$. We now define $\mathcal{C}_{\text{new}}=(\mathcal{C}\setminus \mathcal{C}_b)\cup \mathcal{C}_b'$. Let $I'_{\text{new}}=I'\cup\{e_1\}\setminus \{b\}$ be an independent subset of $\cup_{f=j}^k S'_{f}$. Moreover, each $C_s\in \mathcal{C}_{\text{new}}\setminus \mathcal{C}_b$, is a circuit in $I'_{\text{new}}\cup\{e_s\}$ and each $C'_{s}\in \mathcal{C}_{\text{new}}\cap \mathcal{C}'_b$ is a circuit in $I'_{\text{new}}\cup \{e_s\}$. Thus, $\mathcal{C}_{\text{new}}$ and $I'_{\text{new}}$ satisfy the premise of the inductive hypothesis. However, $\card{\mathcal{C}_{\text{new}}}=t$ and $\card{\cup_{C\in \mathcal{C}_{\text{new}}}C\cap (I'_{\text{new}}\setminus S'_j)}<t$, thus we can conclude by inductive hypothesis that there must be a circuit in $S'_j$. This contradicts the assumption that $S'_j$ is an independent set and concludes the proof of the lemma.
\end{proof}
With this, we are prepared to show the proof of \Cref{lem:greedymerge} about greedily combining the common independent sets of intervals of $\varepsilon^{-1}$-spread matroids. 
\begin{proof}[Proof of \Cref{lem:greedymerge}]
    Let us consider the elements $\{e_1,\cdots,e_l\}$ that were excluded in $I'_j$. Define $I^j_{\text{merge}}=I_{\text{merge}}\cap \cup_{s=1}^j I'_s$. An element $e_i$ was not included in $I_{\text{merge}}$ since $I^j_{\text{merge}}\cup\{e_i\}$ contained a circuit in either $\matroid'_1$ or $\matroid'_2$. Thus we can assume without loss of generality, that $\{e_1,\cdots,e_{\lceil\frac{l}{2}\rceil}\}$ are spanned by $I^j_{\text{merge}}$ in $\matroid'_1$ (by \Cref{lem_charging}). Let $C_{t}$ be the unique circuit in $I^j_{\text{merge}}\cup\{e_t\}$ for $t\in \lceil\frac{l}{2}\rceil$, then $\card{(\cup_{t=1}^{\lceil\frac{l}{2}\rceil}C_t)\cap (I^{j}_{\text{merge}}\setminus I'_{j})}\geq \lceil\frac{l}{2}\rceil$. Let $S=(\cup_{t=1}^{\lceil\frac{l}{2}\rceil}C_t)\cap (I^{j}_{\text{merge}}\setminus I'_{j})$. To each $e\in S\cap L_{p}$ for $p\geq j+1$, we charge the weight of (at most) two elements from $\{e_1,\cdots, e_l\}$ to $e$. Observe that $w(e)\geq l_{p}$ and for each $e_{i}$, $w(e_i)\leq r_j$. Thus, $e$ gets at most $2\varepsilon^{p-j}\cdot w(e)$ charge from $L_{j}$. Thus, the total charge on $e\in S\cap L_{p}$ is at most:
    \begin{align*}
        \sum_{j<p}2\varepsilon^{p-j}w(e)\leq \frac{2\varepsilon w(e)}{1-\varepsilon} \leq 4\varepsilon w(e).
    \end{align*}
where the last inequality holds when $\varepsilon<\frac{1}{2}$. Thus, we can conclude that each excluded element charges its weight to some element of $I_{\text{merge}}$ and the total charge on each element $e$ of $I_{\text{merge}}$ is at most $4\varepsilon w(e)$. Consequently,
\begin{align*}
w(I_{\text{merge}}) &= \sum_{j = 1}^k w(I'_j) - \sum_{e \in \bigcup_{j=1}^k I'_j \setminus I_{\text{merge}}}w(e)\\
&\geq \sum_{j = 1}^k w(I'_j) - 4 \varepsilon\sum_{e \in I_{\text{merge}}} w(e)\\
& \geq (1 - 4\varepsilon)\sum_{j = 1}^k w(I'_j)\qedhere
\end{align*}
\end{proof}

We now show the following corollary about $\varepsilon^{-1}$-spread matroids.

\begin{corollary}\label{cor:greedymergeguarantee}
Let $\matroid_1',\matroid_2'$ be an instance of weighted matroid intersection that is $\varepsilon^{-1}$-spread. For $j\in [k]$, let $L_j=[l_j,r_j]$ be the corresponding intervals and $I_j'$ be an $\alpha$-approximate common independent set of $\matroid_1'|L_j$ and $\matroid_2'|L_j$.
Suppose $I_{\text{merge}}$ is obtained by greedily combining $I'_j$ for $j\in [k]$. Then, 
\begin{align*}
    w(I_{\text{merge}})\geq \alpha \cdot (1-4\varepsilon)\cdot w(I^*),
\end{align*}
where $I^*$ is the maximum weight common independent set of $\matroid_1',\matroid'_2$. Moreover, the time taken by greedy combine to compute $I_{\text{merge}}$ is at most $O(r\cdot k)$.
\end{corollary}
\begin{proof}
Let $I_{j}^*$ denote the optimal common independent set of $\matroid_1'|L_j$ and $\matroid_2'|L_j$. We can conclude the following. 
\begin{align*}
    w(I_{\text{merge}})&\geq (1-4\varepsilon)\sum_{j=1}^{k}w(I'_j)\\
    &\geq \alpha\cdot (1-4\varepsilon)\cdot \sum_{j=1}^{k}w(I^*_j)\\
    &\geq \alpha\cdot (1-4\varepsilon)\cdot w(I^*).\qedhere
\end{align*}
\end{proof}
\subsubsection{Converting Arbitrary Matroids to $\varepsilon^{-1}$-Spread Matroids.}
\label{subsec:generalmatroids}

We begin by giving the following definition.

\begin{definition}[$\varepsilon^{-1}$-matroids from arbitrary matroids]\label{def:spreadmatroids}
    Suppose $\matroid'_1=(\groundset',\independentsets'_1,w)$ and $\matroid'_2=(\groundset',\independentsets'_2,w)$ is an instance of weighted matroid intersection. Let $\beta=\lceil\varepsilon^{-1}\rceil$. For $i\in [\beta]$ and $l\geq 0$, we let 
    $$\groundset'_{i,l}:=\set{e\in \groundset'\mid (\varepsilon^{-1})^{i+(l-1)\beta+1}\leq w(e)<(\varepsilon^{-1})^{i+l\beta}}.$$ Let $\groundset'_i=\cup_{l} \groundset'_{i,l}$. Define $\matroid'_{1,i}=\matroid'_{1}| \groundset'_i$ and $\matroid'_{2,i}=\matroid'_{2}| \groundset'_i$.
\end{definition}

\begin{observation}\label{obs:spreadmatroidsarespread}
Given an instance $\matroid'_1=(\groundset',\independentsets'_1,w)$ and $\matroid'_2=(\groundset',\independentsets'_2,w)$ of matroid intersection. For $i\in [\beta]$, let $\matroid'_{1,i}$ and $\matroid'_{2,i}$ be defined as above. These matroids have the following properties.
\begin{enumerate}
\item For all $i\in [\beta]$, $\matroid'_{1,i}$ and $\matroid'_{2,i}$ are $\varepsilon^{-1}$-spread. 
\item For all $i\in [\beta]$, for all $l$, $\groundset'_{i,l}$ the ratio between the maximum and minimum weight element is $\gamma_{\varepsilon}$. Moreover, by standard scaling and rounding (\Cref{sec:rescaling}), we can assume that $\matroid'_{1,i}\mid \groundset'_{i,l}$ and $\matroid'_{2,i}\mid\groundset_{i,l}$ are weighted matroids with integer element weights $\set{1,2,\cdots, \gamma_{\varepsilon}}$.
\end{enumerate}
\end{observation}

\begin{lemma}\label{lem:oneindependentsetisgood}
Let $I^*$ be the maximum weight common independent set of matroids $\matroid'_1$ and $\matroid'_2$. For $i\in [\beta]$, let $\matroid'_{1,i}$ and $\matroid'_{2,i}$ be as defined above. Let $I^*_{i}$ be the maximum weight common independent sets of $\matroid'_{1,i},\matroid'_{2,i}$ respectively. Then, there exist $i\in [\beta]$ such that $w(I^*_i)\geq (1-\varepsilon)w(I^*)$. 
\end{lemma}
\begin{proof}
    Let $\mathcal{F}'_i=\groundset'\setminus \groundset'_i$. Observe that $\set{\mathcal{F}'_{i}}_{i\in [\beta]}$ form a partition of $\groundset'$. Consequently, $\set{I^*\cap \mathcal{F}'_{i}}_{i\in [\beta]}$ form a partition of $I^*$. Thus, there exist $i\in [\beta]$ such that $w(I^*\cap \mathcal{F}'_{i})\leq \varepsilon\cdot w(I^*)$. Thus, for such an $i$, $w(I^*\cap \groundset'_i)\geq (1-\varepsilon)w(I^*)$. Observe that $I^*\cap \groundset'_{i}$ is a common independent set of $\matroid'_1,\matroid'_2$. Thus, we can conclude that $w(I^*_{i})\geq w(I^*\cap \groundset'_{i})\geq (1-\varepsilon)w(I^*)$. This concludes the proof.
\end{proof}

\begin{claim2}\label{claim:reductionspreadtoarbitrary}
    Suppose we have an algorithm $\mathcal{A}_w$ which computes an $\alpha$-approximate maximum weight common independent set of $\varepsilon^{-1}$-spread matroids in $T_{w}$ time. Then there is an algorithm $\mathcal{A}'_{w}$ that computes an $\alpha(1-\varepsilon)$-approximate maximum weight independent set of arbitrary matroids in total time $O(T_w\cdot \varepsilon^{-1})$.  
\end{claim2}

\begin{proof}
    Given two arbitrary weighted matroids $\matroid_1'=(\groundset',\independentsets_1',w)$ and $\matroid'_2=(\groundset',\independentsets_2',w)$ we can convert these into $\beta$ pairs of matroids $\set{(\matroid'_{1,i},\matroid'_{2,i})}_{i\in [\beta]}$ (in accordance with \Cref{def:spreadmatroids}) that are $\varepsilon^{-1}$-spread (by Observation~\ref{obs:spreadmatroidsarespread}). We run $\mathcal{A}_w$ on each of these to obtain $(1-\varepsilon)$-approximate maximum weight independent sets $I_i$. We output the one with the highest weight. By \Cref{lem:oneindependentsetisgood} we can conclude that $w(I_i)\geq (1-2\varepsilon) w(I^*)$, where $I^*$ is the maximum weight common independent set of $\matroid'_1$ and $\matroid'_2$.
\end{proof}

\subsubsection{Aspect Ratio Reduction for Matroid Parity.}

The problem of matroid parity (or matroid matching) generalizes both matching and matroid intersection. Here, we are given a matroid $\matroid = (\groundset, \independentsets)$ as well as a collection $E$ of disjoint unordered pairs of elements in $\groundset$, where each pair forms an independent set in $\independentsets$. The aim is to find a disjoint subset of $E$, such that the elements form an independent set of maximum cardinality or weight in the unweighted or weighted version of the problem, respectively. This problem is equivalent to matroid matching, where pairs in $E$ need not be disjoint~\cite{LeeSV13}.

The aspect ratio reduction presented in this section, also extends to the matroid parity problem, reducing the total weight $w(P) := w(e_1) + w(e_2)$ of a pair $P = \{e_1, e_2\} \in E$. Instead of considering individual elements, we now need the weights of all pairs to fall into disjoint weight intervals $[l_1,r_1],[l_2,r_2],\ldots, [l_k, r_k]$ such that $l_{i+1} > r_i\varepsilon^{-1}$. This can be achieved analogous to \ref{def:spreadmatroids}, setting $E_{i,l} := \{ \{e_1, e_2\} \in E \mid (\varepsilon^{-1})^{i +(l-1)\beta + 1} \le w(e_1) + w(e_2) < (\varepsilon^{-1})^{i + l\beta}\}$ and $\groundset_{i,l} := \bigcup_{P \in E_i, l} P$ for $i \in [\beta] $ and $ l \geq 0$. Then we consider the matroid $\matroid_i = \matroid | \groundset_i$ where $\groundset_i := \bigcup_{l \ge 0} \groundset_{i,l}$. 

Now, if we have an approximate matroid parity solution $I_{i,l} \subseteq E_{i,l}$ for each $\matroid_i |  \groundset_{i,l} = \matroid | \groundset_{i,l}$ and the corresponding collection $E_{i,l}$ for $l \geq 0$, we can order the pairs in $E_i := \bigcup_{l \ge 0} E_{i,l}$ according to their weight and greedily combine them starting with the highest weighted pairs to create $I_\text{merge} \subseteq E_i$ such that the elements of pairs in $I_\text{merge}$ form an independent set in  $\matroid_i$. We again get that
\begin{equation*}
    \sum_{P \in I_\text{merge}}w(P) \geq (1- 4\varepsilon) \sum_{l \geq 0} \sum_{S \in I_{i,l}} S.
\end{equation*}
This is due to a charging argument similar to \Cref{lem_charging} and \Cref{lem:greedymerge}. If there are $j$ pairs in $I_{i,l}$, that are not added during the greedy process, then each such pair contains at least one element that forms a circuit in $\matroid_i$ with elements that are part of a higher (or equally) weighted pair in $I_\text{merge}$. Analogous to \Cref{lem_charging}, there have to be at least $j$ elements in pairs in $I_\text{merge}$, the loss can be charged to. These $j$ elements are part of at least $j/2$ pairs in $I_\text{merge}$ and the lower bound for the weight of the greedy solution follows similarly to \Cref{lem:greedymerge}, by charging the weight of at most two discarded pairs per weight interval to each pair in the solution. Finally, the greedy combination of $\alpha$-approximate matroid parity solutions $I_{i,l}$, results in an $\alpha \cdot (1 - 4 \varepsilon)$ approximation of the original problem for some $i$.

\subsection{Refolding Unweighted Matroids}
\label{sec:refolding}
In this section, we show how to extract relevant elements of the original matroids from the unfolded matroid intersection. We begin with the following definition.

\begin{definition}[Refolded Matroids]\label{def:refolding}
Let $\matroid_1=(\groundset,\independentsets_1)$ and $\matroid_2=(\groundset,\independentsets_2)$ be two unweighted matroids obtained from two integer-weighted matroids $\matroid'_1=(\groundset',\independentsets'_1,w), \matroid'_2=(\groundset',\independentsets'_2,w)$ through the process of matroid unfolding as described in \Cref{def:matroidunfolding}. Let $S\in \independentsets_1\cap \independentsets_2$. Let $\groundset'_S\subseteq \groundset'$ be the set $\groundset'_S := \{e \in \groundset' | e_j \in S \text{ for some $j \in [w(e)]$}\}$. We refer to $\matroid'_{S,1} = \matroid'_1 | \groundset'_S$ and $\matroid'_{S,2} = \matroid'_2 | \groundset'_S$ as the refolded matroids.
\end{definition}

The following lemma will guarantee that the restricted matroids resulting from refolding contain a common independent set of large weight.

\begin{lemma}\label{lem:refoldguarantee}
Suppose we have an instance of weighted matroid intersection $\matroid'_{1}=(\mathcal{N}',\mathcal{I}'_1,w),\matroid'_2=(\groundset',\independentsets'_2,w)$ and a set $S \subseteq \groundset$ that is an $\alpha$-approximate maximum cardinality common independent set of the unfolded $\matroid_1, \matroid_2$. Let $\opt'$ be the weight of the maximum weight common independent set of $\matroid'_1$ and $\matroid'_2$. Let $\opt'_S$ be the weight of the maximum weight common independent set of $\matroid'_{S,1}, \matroid'_{S,2}$. Then, $\opt'_S \geq \alpha \cdot \opt'$
\end{lemma}

\begin{proof}
Let $\matroid_{1,S}=(\groundset_S,\independentsets_{1,S})$ and $\matroid_{2,S}=(\groundset_S,\independentsets_{2,S})$ be unfolded versions of $\matroid'_{1,S},\matroid'_{2,S}$. Let $\opt_S$ denote the size of the maximum cardinality common independent set of $\matroid_{1,S},\matroid_{2,S}$. Observe that $S\subseteq \groundset_S$ and $S\in \independentsets_{1,S}\cap \independentsets_{2,S}$. Therefore, we have the following line of reasoning:
\begin{equation*}
    \alpha\cdot\opt'=\alpha\cdot\opt\leq |S|\leq \opt_S=\opt'_S. \qedhere
\end{equation*}
\end{proof}
Finally, we note that the size of the groundset after refolding is reduced as follows.

\begin{observation}\label{obs:sizeofrefoldedsets}
    Suppose we have an instance of weighted matroid intersection $\matroid'_{1}=(\mathcal{N}',\mathcal{I}'_1,w)$ and $\matroid'_2=(\groundset',\independentsets'_2,w)$. Let $r$ be the rank of the intersection of $\matroid'_1,\matroid'_2$. Let $\matroid_1,\matroid_2$ be unfolded versions of $\matroid'_1,\matroid'_2$. Let $S$ be a common independent set of $\matroid_1,\matroid_2$. Then, $\card{\groundset'_S}\leq W\cdot r$.
\end{observation}

\subsection{Maximum Weight Common Independent Set with Additive Error}
\label{sec:mwmadditive}
In this section, we show how to obtain an approximation to the maximum weight common independent set with additive error in $O(\varepsilon W n) $. The algorithm is a weighted version of the auction algorithm of \cite{BlikstadT25}. Their algorithm makes $\varepsilon n$ weight adjustments at once, resulting in fewer computations of maximum weight bases.

\begin{lemma}\label{lem:weighted_additive}
Let $\mathcal{M}'_a = (\mathcal{N}', \mathcal{I}'_a, w), \mathcal{M}'_b = (\mathcal{N}', \mathcal{I}'_b, w)$ be two weighted matroids where the maximum weight is $W$, $n$ is the size of the ground set and $r$ is the rank of their intersection. For a given $\varepsilon$, Algorithm \ref{alg:mi_in_level} returns a set $S \in \mathcal{I}_a \cap \mathcal{I}_b$ such that $w(S) \geq w(S^*) - 3W\varepsilon n$ where $S^*$ is a common independent set of highest weight using $O\left(\frac{n}{\varepsilon^2}\right)$ independence queries.
\end{lemma}

We begin by giving \Cref{alg:mi_in_level}.

\begin{algorithm}
\algorithmicrequire{ Weighted matroids $\mathcal{M}'_a = (\mathcal{N}', \mathcal{I}'_a, w), \mathcal{M}'_b = (\mathcal{N}', \mathcal{I}'_b, w)$}, parameter $\varepsilon > 0$\\
	\algorithmicensure{ Common independent set $S$ }
	%\caption{\textsc{WeightedIndSet}()}
	\begin{algorithmic}[1]
        \State Initialize $p(e) \gets 0$, $w_a(e) \gets w(e)$, $w_b(e) \gets w(e)$ for all $e \in \groundset'$.
        \State $S_a \gets$ maximum weight base of $\matroid'_a$ with respect to $w_a$, $S_b \gets$ maximum weight base of $\matroid'_b$ with respect to $w_b$.
        \While{\textbf{true}\label{line:whileTrueLoop}}{
            \State $X \gets \{ e \in S_a \setminus S_b\; |\; p(e) < 2\lceil\varepsilon^{-1}\rceil\}$.
            \If{$|X| \leq \varepsilon n$}
                \State return $S := S_a \cap S_b$.
            \EndIf
            \For{$e \in X$}{
                \State $p(e) \gets p(e) + 1$.
                \If{$w_a(e) + w_b(e) = w(e)$}
                    \State $w_b(e) \gets w_b(e) + \varepsilon (1 - \varepsilon) w(e)$,
                \Else
                    \State $w_a(e) \gets w_a(e) - \varepsilon (1 - \varepsilon) w(e)$.
                \EndIf
            }\EndFor
            \For{ $i \in \{a,b\}$}
                \State $S_i \gets $ maximum weight base of $(\groundset', \mathcal{I}'_i,w_i)$, tie-breaking prefers elements in old $S_i$.
            \EndFor
        }\EndWhile
	\end{algorithmic}
     \caption{An Additive Approximation for Weighted Matroid Intersection}
     \label{alg:mi_in_level}
\end{algorithm}

Since the algorithm is only a slightly modified version of Algorithm 1 from \cite{BlikstadT25}. In particular the following observation transfers directly.

\begin{observation}[Claim 3.1 from \cite{BlikstadT25}]\label{obs:zeroWeight}
Throughout the algorithm all elements $ e \in S_b \setminus S_a$ have $b$-weight $w_b(e) = 0$.
\end{observation}

Next, we want to utilize this observation to argue about the guarantees of \Cref{alg:mi_in_level}. Let us first introduce the definition of a \textbf{weight-splitting}.

\begin{definition}[Weight-splitting]
    Let $\matroid'_1 = (\groundset', \mathcal{I}'_1, w),\matroid'_2 = (\groundset', \mathcal{I}'_2, w)$ be two weighted matroids with weight function $w: \groundset' \rightarrow \mathbb{R}^{\geq 0}$. We call two functions $w_1: \groundset' \rightarrow \mathbb{R}^{\geq 0}$, $w_2: \groundset' \rightarrow \mathbb{R}^{\geq 0}$ a weight splitting of $w$ if $w_1(e) + w_2(e) = w(e)$ for all $e \in \groundset'$
\end{definition}
Now, we restate the following theorem we will utilize in the proof.
\begin{theorem}[Weight-splitting Theorem \cite{Frank81}]\label{thm:weight_splitting}
The maximum weight with respect to weight function $w$ of a common independent set of matroids $\matroid'_1$ and $\matroid'_2$ is given by 
\begin{equation*}
\min_{\substack{w_1,w_2\\ \text{weight-splitting of } w}} \max_{S_1 \in \mathcal{I}_1,S_2 \in \mathcal{I}_2} w_1(S_1) + w_2(S_2).
\end{equation*}
\end{theorem}
Finally, we are ready to show \Cref{lem:weighted_additive}.
\begin{proof}
Let $S^*$ be a maximum weight common independent set of matroids $\mathcal{M}'_a$, $\mathcal{M}'_b$. By \Cref{thm:weight_splitting} we get that
\begin{align*}
w(S^*) = \min_{\substack{w_1,w_2\\ \text{weight-splitting of } w}} w_1(S_1^*) + w_2(S^*_2)
\end{align*}
for $S^*_1$ and $S^*_2$ being maximum weight bases of $(\mathcal{N}', \mathcal{I}'_a, w_1)$ and $(\mathcal{N}', \mathcal{I}'_b, w_2)$ respectively. For any $e \in \groundset'$ we have $w_a(e) + w_b(e) \in \{w(e), w(e) + \varepsilon(1- \varepsilon)w(e)\}$ throughout the algorithm and at termination. Let us in the following focus on the state of $X$, $w_a$, $w_b$ and the corresponding maximum weight bases $S_a$ and $S_b$ at the end of the execution. We can create a valid weight splitting $w_a$, $w'_b: \groundset' \rightarrow \mathbb{R}^{\geq 0}$ from $w_a$ and $w_b$ as follows:
\begin{equation*}
    w'_b(e):=\begin{cases}
w_b(e) - \varepsilon(1-\varepsilon)w(e)& \text{when } w_a(e) + w_b(e) \neq w(e),\\
w_b(e) & \text{otherwise}.
\end{cases} 
\end{equation*}

By \Cref{thm:weight_splitting} for a maximum weight base $S'_b$ of $(\groundset', \mathcal{I}'_b, w'_b)$ we get
\begin{align*}
    w(S^*) &\leq w_a(S_a) + w'_b(S'_b)\\
    &\text{(By \Cref{thm:weight_splitting})}\\
    &\leq w_a(S_a) + w_b(S'_b)\\
    &\text{(Since for all elements $w'_b(e) \leq w_b(e)$)}\\
    &\leq  w_a(S_a) + w_b(S_b)\\
    &\text{(Since $S_b$ is maximum weight basis of $\matroid'_b$)}\\
    &= \sum_{e \in S_a \cap S_b} (w_a(e) + w_b(e)) + \sum_{e \in S_b \setminus S_a} w_b(e) + \sum_{e \in S_a \setminus S_b} w_a(e)\\
    &= \sum_{e \in S_a \cap S_b} (w_a(e) + w_b(e))  + \sum_{e \in (S_a \setminus S_b) \cap X} w_a(e) +  \sum_{e \in (S_a \setminus S_b) \setminus  X} w_a(e)\\
    &\text{(By Observation~\ref{obs:zeroWeight} $w_b(e) = 0$ for all $e \in S_b \setminus S_a$)}\\
    &\leq \sum_{e \in S_a \cap S_b} (w_a(e) + w_b(e))  + \varepsilon Wn +  \sum_{e \in (S_a \setminus S_b) \setminus  X} w_a(e)\\
    &\text{(Since $|(S_a \setminus S_b) \cap X| = |X| \leq  \varepsilon n, w_a(e) \leq w(e) \leq W$ for $e \in \groundset'$)}\\
    &\leq \sum_{e \in S_a \cap S_b} (w_a(e) + w_b(e))  + 2 \varepsilon Wn \\
    &= w(S) + 3 \varepsilon W n,
\end{align*}
where it remains to reason about the last inequality. Consider an element $e \in (S_a \setminus S_b) \setminus  X$. It must be the case that $p(e) > 2\lceil\varepsilon^{-1}\rceil$ and $w_a(e)$ was decreased by $\varepsilon (1-\varepsilon)w(e)$ at least $\varepsilon^{-1}$ times during the execution. Hence, at the end of the execution we have $w_a(e) \leq w(e) - \varepsilon^{-1} \varepsilon (1-\varepsilon)w(e) = \varepsilon w(e) $. Hence $\sum_{e \in (S_a \setminus S_b) \setminus  X} w_a(e) \leq \sum_{e \in (S_a \setminus S_b) \setminus  X} \varepsilon w(e) \leq \varepsilon W n$. Finally, to get 
\begin{equation*}
    w(S^*) \leq w(S) + 3\varepsilon W n 
\end{equation*}
we need to show that $\sum_{e \in S_a \cap S_b} (w_a(e) + w_b(e)) \leq w(S) + \varepsilon W n$. This is true because for each $e \in \groundset'$ we have $w_a(e) + w_b(e) \leq w(e) + \varepsilon (1 - \varepsilon) w(e) \leq  w(e) + \varepsilon W$.

Finally, it remains to analyze the number of queries needed. As in the original algorithm, every iteration of the while loop $p(\groundset') := \sum_{e \in \groundset'} p(e)$ is increased by at least $\varepsilon n$. Since initially $p(\groundset') = 0$ and it is upper bounded by $2n\lceil\varepsilon^{-1}\rceil$. Hence, the total number of queries is in $O\left( \frac{n}{\varepsilon^2}\right)$.
\end{proof}

\subsection{Putting Things Together}

In this section, we state and prove our main reduction in the classical query model. We first state the reduction.

\begin{theorem}\label{lem:staticreduction}
   Let $\mathcal{A}_u$ be an algorithm that computes an $\alpha$-approximate maximum cardinality common independent set using $T_{u}(n,r)$ oracle queries, then there is an algorithm $\mathcal{A}_w$ that computes an $\alpha(1-\varepsilon)$-approximation to the maximum weight common independent set using $O((T_{u}(n\gamma_{\varepsilon},r\gamma_{\varepsilon})+\frac{\gamma^3_{\varepsilon}\cdot r}{\varepsilon^3})\log W)$ oracle queries. Moreover, if $\mathcal{A}_u$ is deterministic, then so is $\mathcal{A}_w$ and if $\mathcal{A}_u$ uses independence queries, then the same is true for $\mathcal{A}_w$.
\end{theorem}

To prove the above reduction, we will first prove the following intermediate results. 

\begin{lemma}\label{red:staticaspectratio}
    Let $\mathcal{A}_w$ be an algorithm that computes an $\alpha$-approximate maximum weight common independent set of two integer-weighted matroids using $T(n,r,W)$ queries. Then, there is an algorithm $\mathcal{A}'_w$ that computes an $\alpha(1-\varepsilon)$-approximate maximum weight common independent set of two matroids with an arbitrary weight function using $O(\frac{(T(n,r,\gamma_{\varepsilon
    })+r)\cdot\log W}{\varepsilon})$ queries.
\end{lemma}
\begin{proof}
    Suppose we have $\matroid'_1=(\groundset',\independentsets'_1,w)$ and $\matroid'_2=(\groundset',\independentsets'_2,w)$ be weighted matroids with $w:\groundset'\rightarrow \mathbb{R}^{\geq 0}$. By \Cref{def:spreadmatroids}, we can convert these into $\varepsilon^{-1}$-spread matroids $\matroid'_{1,i},\matroid'_{2,i}$ for $i\in[\beta]$ such that \Cref{lem:oneindependentsetisgood} holds. 
    
    The algorithm $\mathcal{A}'_w$ will proceed as follows. For each $i\in [\beta]$ and $l\geq 0$, it will run $\mathcal{A}_w$ on $\matroid'_{1,i}\mid \groundset'_{i,l}$ and $\matroid'_{1,i}\mid \groundset'_{i,l}$. Let $\opt'_{i,l}$ denote the weight of the maximum weight independent set of these matroids. Let $\opt'_{i}$ denote the maximum weight independent set of $\matroid'_{i,1},\matroid'_{i,2}$. Observe that by Observation~\ref{obs:spreadmatroidsarespread}, these matroids have integer weights in $\set{1,2,\cdots,\gamma_{\varepsilon}}$. Thus, one can compute $\alpha$-approximate maximum weight common independent sets $I_{i,l}$ for each of these matroids using $\mathcal{A}_w$ using total number of queries $\frac{T(n,r,\gamma_{\varepsilon
    })\cdot\log W}{\varepsilon}$. Note that $|I_{i,l}|\leq r$ for all $i\in [\beta]$ and $l\geq 0$. Therefore, for each $i\in [\beta]$, we can greedily combine $I_{i,l}$ to obtain the independent set $I_{i}$. This takes $O(\frac{r}{\varepsilon}\log W)$ queries. Let $I_{i^*}$ be the obtained set with the highest weight. We have,
    \begin{align*}
        w(I_{i^*})&\geq (1-4\varepsilon)\sum_{l\geq 0} w(I_{i^*,l})\\
        &\text{(By \Cref{lem:greedymerge})}\\
        &\geq (1-4\varepsilon)\cdot \alpha\sum_{l\geq 0}\opt'_{i^*,l}\\
        &\geq (1-4\varepsilon)\cdot \alpha\cdot\opt'_{i^*}\\
        &\geq (1-4\varepsilon)\alpha\cdot \opt'\\
        &\text{(By \Cref{lem:oneindependentsetisgood})}
    \end{align*}
This concludes the proof.
\end{proof}

\begin{lemma}\label{lem:intweightedtounweighted}
     Let $\mathcal{A}_u$ be an algorithm that computes an $\alpha$-approximate maximum cardinality common independent set using $T_u(n,r)$ queries. Then there is an algorithm $\mathcal{A}_w$ that computes an $\alpha$-approximation to the maximum weight common independent set for integer-weights using $O(T_{u}(nW,rW)+\frac{W^3r}{\varepsilon^2})$ queries. 
\end{lemma}
\begin{proof}
    Given such matroids $\matroid'_1=(\groundset',\independentsets'_1,w)$ and $\matroid'_2=(\groundset',\independentsets'_2,w)$ with $W$ as the ratio between weights, we first unfold them to create matroids $\matroid_1=(\groundset,\independentsets_1)$ and $\matroid_2=(\groundset,\independentsets_2)$. Note that $\card{\groundset}\leq W\card{\groundset'}$ and similarly, $\opt\leq rW$ by Observation~\ref{obs:unfoldedmatroidsize}. Thus, in $T(nW,rW)$ queries, we can compute an $\alpha$-approximate maximum cardinality common independent set $S$ of $\matroid_1,\matroid_2$. Subsequently, we obtain refolded matroids $\matroid'_{1,S}$ and $\matroid'_{2,S}$. Note that $|\groundset'_S|\leq W\cdot r$ by \Cref{lem:refoldguarantee}. Let $S'$ be the maximum weight common independent set of $\matroid'_{1,S}$ and $\matroid'_{2,S}$. We can run \Cref{alg:mi_in_level} with $\delta=\frac{\varepsilon}{W^2}$ as a precision parameter to get an common independent set $T'$ of $\matroid'_{1,S}$ and $\matroid'_{2,S}$ such that
    \begin{align*}
        w(T')&\geq w(S')-2\delta W|\groundset'_S|\\
        &\text{(By \Cref{lem:weighted_additive})}\\
        &\geq w(S')-2\delta W^2 r\\
        &\text{(Since $|\groundset'_S|\leq r\cdot W$)}\\
        &\geq (1-\varepsilon)w(S')\\
        &\text{(Since weights are integral $w(S')\geq r$)}
    \end{align*}
This requires $O(\frac{W^3r}{\varepsilon^2})$ queries by \Cref{lem:weighted_additive}.
\end{proof}

We now show the proof of the reduction.

\begin{proof}[Proof of \Cref{lem:staticreduction}]
    The proof directly follows by \Cref{red:staticaspectratio} and \Cref{lem:intweightedtounweighted}.
\end{proof}

\section{Applications}
\label{sec:applications}

In this section, we delve into the implications of our reduction to the classical, streaming and communication complexity settings. Our contributions are summarized in the form of a table. 

\begin{table}[H]
\centering
	\caption{Summary of Prior and Our Results on Weighted Matroid Intersection}
		\begin{tabular}{c c c c c}
		\hline
		Setting	& Prior Result & Our Result\\
		\hline

\makecell{Indep. Queries \\ (Randomized)} & \makecell{$O(\frac{n\log n}{\varepsilon}+\frac{r^{1.5}}{\varepsilon^4})$\\\cite{Quanrud24}} & \makecell{$O(\gamma_{\varepsilon}(n\log n+r\log^3n +\gamma_{\varepsilon}^2r)\log W)$\\\Cref{lem:linearmatroid}}  \\
  \hline
\makecell{Rank Queries\\ (Deterministic)} & \makecell{$\tilde{O}(\frac{n \sqrt{r}}{\varepsilon^2})$\\\cite{Quanrud24}(\cite{ChekuriQ16} + \cite{Blikstad21}} & \makecell{$O(\gamma_{\varepsilon} n\log n \log W +r\gamma^3_{\varepsilon}\log W)$\\\Cref{lem:rankqueriesdeterministic}} \\
\hline
\makecell{Semi-streaming\\(Multipass)\\ $(1-\varepsilon)$-approx.} & \makecell{Space: \\$O((r_1+r_2)\poly(\varepsilon^{-1})\log W)$ \\ Passes: \\ $O(\poly(\varepsilon^{-1})\log n)$ \\ \cite{Quanrud24}} & \makecell{Space: \\ $O(\varepsilon^{-1}(r_1+r_2)\gamma_{\varepsilon}\log W)$ \\ Passes: \\ $O(\varepsilon^{-1})$ \\\Cref{lem:detstreaming}} \\
\hline
\makecell{Semi-streaming \\ (One-Pass)} & \makecell{$(0.5-\varepsilon)$-approximate \\ $O(\frac{(r_1+r_2)\cdot \log W}{\varepsilon})$\\
\cite{GargJS21}} & \makecell{ $(0.5-\varepsilon)$-approximate \\
Space: \\ $O(\gamma_{\varepsilon} r\log W)$  \\ \Cref{lem:onepassstreaming}} \\
\hline
\makecell{Semi-Streaming\\ Random-Order \\ (One-Pass)} & \makecell{No Prior Results} & \makecell{$(\frac{2}{3}-\varepsilon)$-approximation\\ Space: \\ $O(r \log n \log (\min\{r_1,r_2\}) \gamma_{\varepsilon}\log W)$  \\
\Cref{lem:randomlyorderedweighted}} \\
\hline
\makecell{One-Way CC} & \makecell{No Prior Results} & \makecell{Space: \\ $O(\gamma_{\varepsilon} r \log W)$ \\ \Cref{lem:commcomplexityresult}} \\
\hline
\end{tabular}
\label{table:summary of results}
\end{table}

In addition to these applications, also obtain some results for the case of bipartite weighted $b$-matching in the MPC, parallel shared-memory work-depth, and distributed blackboard models. We elaborate on these in \Cref{sec:bmatching}. We note that \cite{HuangS22} obtained a one-pass algorithm for $(2-\delta)$-approximate maximum weight $b$-matching in random-order streams with space $O(\max(|M_G|, |V|) \cdot\poly(\log(m), W, \frac{1}{\varepsilon}))$, where $V$ is the vertex set of $G$, $|M_G|$ refers to the cardinality of the optimal $b$-matching and $m$ is the number of edges in the graph. For the specific case of bipartite graphs, our result on random-order streams improves on theirs. 

Finally, we clarify some aspects of \Cref{table:summary of results}. In the rows which say ``No Prior Results'', we mean specifically for the case of general weighted matroid intersection. For these models, there has been a lot of progress for the special case of weighted bipartite matching. In particular, on the question of multi-pass semi-streaming algorithm for maximum weight matching problem, there is a large body of work and we refer interested reader to \cite{Assadi24} for a discussion about this problem. For random-order streams, the state of the art for weighted matching is due to \cite{AssadiB21,BernsteinDL21}. Finally, for one-way communication complexity, the best bounds are due to \cite{GoelKK12,AssadiB18,BernsteinDL21}. We also refer the reader to \cite{AzarmehrB23} for the best known result on the ``robust'' version of this problem.

\subsection{Static Setting}

Our main application of \Cref{lem:staticreduction} is the following result, which is the first linear time \emph{independence query algorithm} for $(1-\varepsilon)$-approximate maximum weight matroid intersection.

\begin{lemma}\label{lem:linearmatroid}
    There is a {\bf randomized} algorithm $\mathcal{A}_w$, which on input two matroids $\matroid'_1=(\groundset',\independentsets_1',w)$ and $\matroid_2'=(\groundset',\independentsets_2',w)$ outputs a $(1-\varepsilon)$-approximation to the maximum weight common independent set using $O(\gamma_{\varepsilon}\cdot (n\log n+r\log^3n +\gamma_{\varepsilon}^2 r )\cdot \log W)$ independence queries. 
\end{lemma}

Our result is obtained by applying \Cref{lem:staticreduction} to the following recent result of \cite{BlikstadT25}. The proof is straightforward, so we omit it. 

\begin{lemma}[\protect{\cite[Theorem 1.2]{BlikstadT25}}]
    There is a {\bf randomized} algorithm that on input matroids $\matroid_1=(\groundset,\independentsets_1)$ and $\matroid_2=(\groundset,\independentsets_2)$ outputs $S\in \independentsets_1\cap \independentsets_2$ such that $|S|\geq (1-\varepsilon)r$ using $O(\frac{n\log n}{\varepsilon}+\frac{r\log^3n}{\varepsilon^5})$ independence oracle queries. 
\end{lemma}

Our second application builds on the following result by \cite{ChakrabartyLSSW19}. 

\begin{lemma}
    There is a {\bf deterministic algorithm} that on input matroids $\matroid_1=(\groundset,\independentsets_1)$ and $\matroid_2=(\groundset,\independentsets_2)$ outputs a $(1-\varepsilon)$-approximate maximum cardinality common independent set using $O(\frac{n\log n}{\varepsilon})$ {\bf rank queries}.
\end{lemma}

Combining \Cref{lem:staticreduction} and the above lemma, we can obtain the following result.

\begin{lemma}\label{lem:rankqueriesdeterministic}
    There is a {\bf deterministic algorithm} that on input matroids $\matroid'_1=(\groundset',\independentsets'_1,w)$ and $\matroid'_2=(\groundset',\independentsets'_2,w)$ outputs a $(1-\varepsilon)$-approximate maximum weight common independent set using $O(\gamma_{\varepsilon}\cdot n\cdot \log n\cdot \log W +r\cdot \gamma^3_{\varepsilon}\cdot \log W)$ {\bf rank queries.}
\end{lemma}

We now show how to prove \Cref{lem:staticreduction}.

\subsection{Streaming Setting}
\label{sec:streaming}

In the \emph{streaming model} the goal is to solve a problem over a stream of input while using a limited amount of memory. The algorithms may read the input once or multiple times. In the former case, they are called \emph{one-pass} and in the latter case, they are referred to as \emph{multi-pass}. We obtain the following reduction in the streaming setting. 

\begin{lemma}\label{lem:streamingreduction}
    Let $\mathcal{A}_u$ be an algorithm that computes an $\alpha$-approximation to the maximum cardinality common independent set of matroids $\matroid_1,\matroid_2$ using $S_{u}(n,r_1,r_2,r)$ space and $p$ passes. Then, there is an algorithm $\mathcal{A}_w$ that computes an $\alpha(1-\varepsilon)$-approximation to the maximum weight common independent set of $\matroid'_{1},\matroid'_2$ using $O(\frac{S_{u}(\gamma_{\varepsilon}n,\gamma_{\varepsilon}r_1,\gamma_{\varepsilon}r_2,\gamma_{\varepsilon}r)\log W}{\varepsilon})$ space and $p$ passes. 
\end{lemma}

Our main application of \Cref{lem:streamingreduction} are the following two results.

\begin{lemma}\label{lem:detstreaming}
    There is a deterministic streaming algorithm, which on input two matroids $\matroid'_1=(\groundset',\independentsets'_1,w)$ and $\matroid'_2=(\groundset',\independentsets'_2,w)$ outputs a $(\frac{2}{3}-\varepsilon$)-approximation to the maximum cardinality common independent set using $O(\varepsilon^{-1}(r_1+r_2)\gamma_{\varepsilon}\log W)$ space and $O(\varepsilon^{-1})$ passes. 
\end{lemma}

This result is obtained by combining \Cref{lem:streamingreduction} with the following result by \cite{Terao24}.

\begin{lemma}[\cite{Terao24}]
There is a deterministic streaming algorithm, which on input two matroids $\matroid_1=(\groundset,\independentsets_1,)$ and $\matroid_2=(\groundset,\independentsets_2)$ outputs a $(\frac{2}{3}-\varepsilon$)-approximation to the maximum cardinality common independent set using $O(r_1+r_2)$ space and $O(\varepsilon^{-1})$ passes. 
\end{lemma}

Our second application is the following result for $(0.5-\varepsilon)$-approximate maximum weight independent set. Prior to this, \cite{GargJS21} obtained an algorithm with the same approximation ratio that had a space complexity of $O(\frac{(r_1+r_2)\cdot\log W}{\varepsilon})$.

\begin{lemma}\label{lem:onepassstreaming}
    There is a streaming algorithm $\mathcal{A}_w$ which on input matroids $\matroid'_{1}=(\groundset', \independentsets'_1,w)$ and $\matroid'_{2}=(\groundset', \independentsets'_2,w)$ outputs a $(0.5-\varepsilon)$-approximate maximum weight common independent set in $O(\gamma_{\varepsilon}\cdot r\cdot \log W)$ space and one pass. 
\end{lemma}

The above algorithm is obtained by combining the greedy algorithm for unweighted matroid intersection, which obtains a $0.5$-approximation in $O(r)$ space and one pass, with \Cref{lem:streamingreduction}.

We now show \Cref{lem:streamingreduction} using the following lemma.

\begin{lemma}\label{lem:streamingaspectratioreduction}
    Suppose there is a streaming algorithm $\mathcal{A}_w$ which on input integer-weighted matroids $\matroid'_1=(\groundset',\independentsets'_1,w)$ and $\matroid'_2=(\groundset',\independentsets'_2,w)$ with weight ratio $W$ gives an $\alpha$-approximation to the maximum weight common independent set in space $S(n,r_1,r_2,r,W)$ and $p$ passes. Then, there is a streaming algorithm $\mathcal{A}'_w$ which on input arbitrary weight matroids with weight ratio $W$ outputs an $\alpha(1-\varepsilon)$-approximation to the maximum weight common independent set in space $O(S(n,r,r_1,r_2,\gamma_{\varepsilon})\log W)$ space and $p$ passes. 
\end{lemma}
\begin{proof}[Proof Sketch]
    The algorithm $\mathcal{A}_w$ proceeds as follows. First, it applies \Cref{def:spreadmatroids} to $\matroid'_{1},\matroid'_{2}$ to obtain spread matroids $\matroid'_{1,i},\matroid'_{2,i}$ for $i\geq 0$ for $i\in [\beta]$. We apply $\mathcal{A}_w$ to the pair $\matroid'_{1,i}\mid \groundset'_{i,l}$ and $\matroid'_{2,i}\mid \groundset'_{i,l}$ (by Observation~\ref{obs:spreadmatroidsarespread}, these matroids satisfy the requirements for $\mathcal{A}_w$)and obtain the independent sets $I_{i,l}$ for $i\in [\beta]$ and $l\geq 0$. Since all independent sets are in the memory, we can combine all $I_{i,l}$ for a fixed $i$, and then take the heaviest independent set across all $i\in [\beta]$. Thus, the total space taken by our algorithm is $O(S_{u}(n,r_1,r_2,r,\gamma_{\varepsilon})\cdot \varepsilon^{-1}\cdot \log W)$ and the number of passes is $p$.
\end{proof}

\begin{remark}\label{remark:generalityreduction}
    The reduction of \Cref{lem:streamingaspectratioreduction} also applies to random order streams. 
\end{remark}

\begin{proof}[Proof of \Cref{lem:streamingreduction}]
It is sufficient to show an $S_{u}(nW,r_1W,r_2W,r)$-space and $p$-pass streaming algorithm $\mathcal{B}_w$ for computing an $\alpha$-approximate maximum weight common independent set for matroids $\matroid'_1,\matroid'_2$, that have integral element weights in $\{1,2,\cdots, W\}$. This is because we can apply \Cref{lem:streamingaspectratioreduction} to $\mathcal{B}_w$ to obtain $\mathcal{A}_w$ with the desired space and pass complexity. 

So we go on to show how $\mathcal{B}_w$ proceeds. When an element $e$ with weight $w(e)$ arrives, it feeds elements $e_1,\cdots, e_{w(e)}$ of the unfolded matroids $\matroid_1=(\groundset,I_{1})$ and $\matroid_2=(\groundset,I_{2})$ to $\mathcal{A}_u$. By Observation~\ref{obs:unfoldedmatroidsize}, we have, $\card{\groundset}\leq \card{\groundset'}\cdot W$ and $\rk_1(\matroid_1)\leq W\cdot \rk'_1(\matroid'_1)$ and $\rk_2(\matroid_2)\leq W\cdot \rk'_2(\matroid'_2)$. Thus, in space $S(nW,r_1W,r_2W,r)$ the algorithm $\mathcal{A}_u$ computes an $\alpha$-approximate maximum cardinality common independent set of $\matroid_1,\matroid_2$ denoted $S$. We are now left with the task of extracting an $\alpha$-approximate maximum weight common independent set of $\matroid'_1,\matroid'_2$ from $S$.
In order to do this, we proceed as follows. At the end of the stream (when we are offline), we create the matroids $\matroid'_{S,1},\matroid'_{S,2}$ as described in \Cref{def:refolding}. By \Cref{lem:refoldguarantee}, the maximum weight common independent set of $\matroid'_{S,1}$ and $\matroid'_{S,2}$ contains an $\alpha$-approximate maximum weight common independent set of $\matroid'_{1}$ and $\matroid'_2$. We can compute the former offline. This completes the proof.
\end{proof}

\subsection{Communication Complexity Setting}
\label{sec:communicationcomplexity}

\paragraph{Model Definition.} In the one-way two-party communication complexity model, Alice and Bob each have some portion of the input, and the goal is to compute some function of the entire input. Alice can talk to Bob, but Bob cannot talk back to Alice; all communication flows in one direction. 
Understanding problems in the one-way two-party model is often seen as a first step to understanding them in more difficult models such as the streaming model, and thus the communication complexity model has received a lot of recent attention \cite{GoelKK12,AssadiB18,AssadiB21b}.

Given a weighted matroid intersection instance $\matroid'_1=(\groundset',\independentsets'_1,w)$ and $\matroid'_2=(\groundset',\independentsets'_2,w)$ in the one-way two-party communication model, Alice and Bob are given $\groundset'_{A}\subset \groundset'$ and $\groundset'_{B}=\groundset'\setminus\groundset'_{A}$ respectively. The goal is for Alice to send a small message to Bob so that Bob can output a good approximation to the maximum weight common independent set of $\matroid'_1,\matroid'_2$. In this setting, we prove the following reduction. 

\begin{lemma}\label{lem:reductioncommcomplexity}
  Suppose we have a protocol $\mathcal{P}_u$ in the one-way two-party communication complexity model that computes an $\alpha$-approximation to the maximum cardinality common independent set of matroids $\matroid_1,\matroid_2$ using $C_{u}(n,r)$ bits of communication. Then, there is a protocol $\mathcal{P}_w$ in the one-way two-party communication complexity model that computes an $\alpha(1-\varepsilon)$-approximation to the maximum weight independent set of two matroids $\matroid_1',\matroid_2'$. The number of bits used by $\mathcal{P}_w$ is $O(C_{u}(\gamma_{\varepsilon}n,\gamma_{\varepsilon}r)\cdot \log W)$. 
\end{lemma}

As an application of this reduction, we obtain the following result, which extends the result of \cite{HuangS24} to the weighted setting.

\begin{lemma}\label{lem:commcomplexityresult}
    There exists a one-way communication protocol, that given $\varepsilon>0$, computes a $(\frac{2}{3}-\varepsilon)$-approximation to the maximum weight common independent set using a message size of $O(\gamma_{\varepsilon}\cdot r\cdot \log W)$ from Alice to Bob, where $r$ is the rank of the maximum cardinality common independent set.
\end{lemma}

The above result is obtained by combining \Cref{lem:reductioncommcomplexity} with the following result of \cite{HuangS24}.

\begin{lemma}
    There exists a one-way communication protocol that, given any $\varepsilon>0$, computes a $(2/3-\varepsilon)$-approximation to the maximum matroid intersection problem using a message of size $O(\frac{r}{\varepsilon})$ from Alice to Bob.
\end{lemma}

We now show the proof of \Cref{lem:reductioncommcomplexity}. The proof will be via the following aspect ratio reduction.

\begin{lemma}\label{lem:communicationaspectratioreduction}
  Suppose there exists a one-way communication protocol, which computes an $\alpha$-approximation to the maximum weight common independent set using a message of size $C(n,r,W)$ (where the elements of the ground set have weights in $[W]$). Then there exists a protocol to compute an $\alpha(1-\varepsilon)$-approximate maximum weight common independent set of matroids with weights in $\mathbb{R}^{\geq 0}$ using a message of size $O(\frac{C(n,r,\gamma_{\varepsilon})\log W}{\varepsilon})$.
\end{lemma}

The proof is similar to \Cref{lem:streamingaspectratioreduction}.

\begin{proof}[Proof Sketch]
    Let $\mathcal{P}$ be the protocol mentioned in the premise of the lemma. Let $\mathcal{P}'$ be the protocol we desire. The protocol $\mathcal{P}'$ proceeds as follows. First, it applies \Cref{def:spreadmatroids} to $\matroid'_1,\matroid'_2$ to obtain spread matroids $\matroid'_{1,i},\matroid'_{2,i}$ for all $i\in [\beta]$. Then, we apply $\mathcal{P}$ to $\matroid'_{1,i}\mid \groundset'_{i,l}$ and $\matroid'_{2,i}\mid \groundset'_{i,l}$ for $l\geq 0$. The protocol $\mathcal{P}'$ sends messages corresponding to each of the matroid pairs to Bob, who then uses them to compute common independent sets $I_{i,l}$ for $i\in [\beta]$, $l\geq 0$. By correctness of $\mathcal{P}$, $I_{i,l}$ is an $\alpha$-approximate maximum weight common independent set of $\matroid'_{1,i}\mid \groundset'_{i,l}$ and $\matroid'_{2,i}\mid \groundset'_{i,l}$. The total size of the messages sent to Bob is $O(C(n,r,\gamma_{\varepsilon})\cdot \varepsilon^{-1}\cdot \log W)$. Once Bob has all the independent sets $I_{i,l}$, he can greedily combine all $I_{i,l}$ for a fixed $i\in [\beta]$ and finally take the heaviest independent set across all $i\in [\beta]$.
\end{proof}

We now show the proof of \Cref{lem:reductioncommcomplexity}.

\begin{proof}[Proof of \Cref{lem:reductioncommcomplexity}]
    It is sufficient to show a $C_{u}(nW,rW)$ message size communication protocol $\mathcal{P}_w$ for computing an $\alpha$-approximate maximum weight independent set for matroids $\matroid'_1,\matroid'_2$, that have element weights in $\set{1,2,\cdots, W}$. This is because we can apply \Cref{lem:communicationaspectratioreduction} to $\mathcal{P}_w$ to obtain a protocol with desired message complexity. 
    So we proceed to show how $\mathcal{P}$ works, given $\mathcal{P}_u$ for $\alpha$-approximate matroid intersection with message complexity $C_{u}(n,r)$. Alice runs $\mathcal{P}_u$ on the unfolded version $\mathcal{N}_A$ of her ground set $\mathcal{N}'_A$. She sends the message to Bob, who uses this message to compute an $\alpha$-approximate common independent set $S$ on the unfolded matroids $\matroid_1,\matroid_2$. After computing $S$, the goal is now to extract an $\alpha$-approximate maximum weight common independent set of $\matroid'_1,\matroid'_2$ from $S$. In order to do this, Bob can refold $S$ as in \Cref{def:refolding} to get matroids $\matroid'_{S,1},\matroid'_{S,2}$. By \Cref{lem:refoldguarantee}, the maximum weight common independent set of $\matroid'_{S,1},\matroid'_{S,2}$ is an $\alpha$-approximate maximum weight independent set of $\matroid'_1,\matroid'_2$. The message complexity of this protocol is $C_{u}(nW,rW)$.
\end{proof}

\subsection{Robust Sparsifiers and Random-Order Streams}
\label{sec:sparsifier}

In this section, we will focus on obtaining robust sparsifiers for weighted matroid intersection and on algorithms for this problem in random-order streams. 

\paragraph{Sparsifiers for Matroid Intersection.} In graph algorithms, ``sparsification" is a tool introduced to convert a dense graph into a sparse graph, while preserving some relevant information from the graph approximately. Inspired by matching sparsifiers, \cite{HuangS24} introduced \emph{density constrained subset}, which is a sparsifier for matroid intersection which has some special properties.

\begin{lemma}[\cite{HuangS24}]
    Let $\matroid_1=(\groundset,\independentsets_1)$ and $\matroid_2=(\groundset,\independentsets_2)$ be a matroid intersection instance. Then, there exists a $S\subseteq \groundset$ (called a density controlled subset or DCS) with the following properties.
    \begin{enumerate}
        \item {\bf(Sparseness)} $|S|=O(\frac{\opt}{\varepsilon})$.
        \item {\bf(Approximation Ratio)} $\matroid_1\mid S$ and $\matroid_2\mid S$ contain a common independent set of size at least $(\frac{2}{3}-\varepsilon)\cdot \opt$.
        \item {\bf(Robustness)}  Suppose $\groundset_A\cup\groundset_B$ be a partition of the groundset $\groundset$. Let $S$ be a DCS of $\matroid_1\mid \groundset_A$ and $\matroid_2\mid \groundset_A$. Then, $\matroid_1\mid S\cup \groundset_B$ and $\matroid_2\mid S\cup \groundset_B$ contain an independent set of size at least $(\frac{2}{3}-\varepsilon)\cdot \opt$.
    \end{enumerate}
\end{lemma}

Our first result is the following weighted version of the sparsifier of \cite{HuangS24}.

\begin{lemma}[Weighted Robust Sparsifier]\label{lem:weightedsparsifier}
    Let $\matroid'_1=(\groundset',\independentsets'_1,w)$ and $\matroid'_2=(\groundset',\independentsets'_2,w)$ be a weighted matroid intersection instance with $w:\groundset'\rightarrow \mathbb{R}^{\geq 0}$. Then, there exists a $\groundset'_S\subseteq \groundset'$ with the following properties (called a weighted degree controlled subset or w-DCS).
    \begin{enumerate}
        \item {\bf (Sparseness)} $|\groundset'_S|=O(r\cdot \gamma_{\varepsilon}\cdot \log W)$.
        \item {\bf (Approximation Ratio)} $\matroid'_1\mid \groundset'_S$ and $\matroid'_2\mid \groundset'_S$ contain a common independent set of weight at least $(\frac{2}{3}-\varepsilon)\cdot \opt'$.
        \item {\bf (Robustness)}  Suppose $\groundset'_A\cup\groundset'_B$ is a partition of $\groundset'$. Let $\groundset'_S$ be a DCS of $\matroid'_1\mid \groundset'_A$ and $\matroid'_2\mid \groundset'_A$. Then, $\matroid'_1\mid \groundset'_S\cup \groundset'_B$ and $\matroid'_2\mid \groundset'_S\cup \groundset'_B$ contain an independent set of weight at least $(\frac{2}{3}-\varepsilon)\cdot \opt'$.
    \end{enumerate}
\end{lemma}

In order to show this, we first show the following intermediate robust sparsifier for integer weights. Subsequently, we will use this to derive a proof for \Cref{lem:weightedsparsifier}.

\begin{lemma}\label{lem:integerweightsparsifier}
    Let $\matroid'_1=(\groundset',\independentsets'_1,w)$ and $\matroid'_2=(\groundset',\independentsets'_2,w)$ be a weighted matroid intersection instance with weights in $\set{1,2,3,\cdots, W}$. Then, there exists a $\groundset'_S\subseteq \groundset'$ with the following properties (called a weighted density constrained subset or w-DCS).
    \begin{enumerate}
        \item {\bf (Sparseness)} We have that $|\groundset'_S|=O(r\cdot W)$.
        \item {\bf (Approximation Ratio)} $\matroid'_1\mid \groundset'_S$ and $\matroid'_2\mid \groundset'_S$ contain a common independent set of size at least $(\frac{2}{3}-\varepsilon)\cdot \opt'$.
        \item {\bf (Robustness)}  Suppose $\groundset'_A\cup\groundset'_B$ is a partition of $\groundset'$. Let $\groundset'_S$ be a DCS of $\matroid'_1\mid \groundset'_A$ and $\matroid'_2\mid \groundset'_A$. Then, $\matroid'_1\mid \groundset'_S\cup \groundset'_B$ and $\matroid'_2\mid \groundset'_S\cup \groundset'_B$ contain an independent set of weight at least $(\frac{2}{3}-\varepsilon)\cdot \opt'$.
    \end{enumerate}
\end{lemma}

\begin{proof}
     Let's start with defining the w-DCS. We take unfolded versions of the matroids, namely, $\matroid_1,\matroid_2$ and let $S$ be their DCS. Let $\groundset'_S\subseteq \groundset'$ be the refolded version of $S$. This will be our candidate w-DCS. We now show it satisfies the three properties.
    \begin{enumerate}
        \item {\bf (Sparseness)} By Observation~\ref{obs:sizeofrefoldedsets}, we can conclude that $|\groundset'_S|\leq W\cdot r$. 
        \item {\bf (Approximation Ratio)} Let $T\subseteq S$ be the maximum common independent set of $\matroid_1\mid S$ and $\matroid_2\mid S$ (and therefore the $(\frac{2}{3}-\varepsilon)$-approximate maximum common independent set of $\matroid_1,\matroid_2$ as $S$ is a DCS of $\matroid_1,\matroid_2$). Suppose matroid pairs $\matroid'_{1,S}$ and $\matroid'_{2,S}$; and $\matroid'_{1,T}$ and $\matroid'_{2,T}$ be obtained by refolding as in \Cref{def:refolding}. By \Cref{lem:refoldguarantee}, the weight of the maximum common independent set of $\matroid'_{1,T}$ and $\matroid'_{2,T}$ has weight at least $(\frac{2}{3}-\varepsilon)\opt'$. This independent set is also a common independent set of $\matroid'_{1,S}$ and $\matroid'_{2,S}$ as $\groundset'_T\subseteq \groundset'_S$.
        \item {\bf (Robustness)} Let $\groundset'_S$ be obtained by computing a DCS $S$ of $\matroid_1\mid \groundset_A$ and $\matroid_2\mid \groundset_A$ and refolding it. Let $\groundset_S\subseteq \groundset $ be obtained by unfolding $\groundset'_S$. Observe that $S\subseteq \groundset_S$. Secondly, by the fact that $S$ is a DCS of $\matroid_1\mid \groundset_A$ and $\matroid_2\mid \groundset_A$, we can conclude that $\matroid_1\mid S\cup \groundset_B$ and $\matroid_2\mid S\cup \groundset_B$ contain an independent set of size at least $(\frac{2}{3}-\varepsilon)\cdot \opt$ which is equal to $(\frac{2}{3}-\varepsilon)\cdot \opt'$. We can also conclude this for $\matroid_1\mid \groundset_S\cup \groundset_B$ and $\matroid_2\mid \groundset_S\cup \groundset_B$. Since these are unfolded versions of $\matroid'_1\mid \groundset'_S\cup\groundset'_B$ and $\matroid'_2\mid \groundset'_S\cup\groundset'_B$, we can additionally conclude that the maximum weight common independent set of $\matroid'_1\mid\groundset'_S\cup \groundset'_B$ and $\matroid'_2\mid\groundset'_S\cup \groundset'_B$ contains a $(\frac{2}{3}-\varepsilon)$-approximation $\opt'$.
    \end{enumerate}
This concludes the proof.
\end{proof}

We now proceed to prove \Cref{lem:weightedsparsifier}. 
\begin{proof}
    We proceed by creating $\beta$ weighted matroid pairs $\matroid'_{i,1}$ and $\matroid'_{i,2}$ as per \Cref{def:spreadmatroids}. Further, we consider matroids $\matroid'_{1,i}\mid \groundset'_{i,l}$ and $\matroid'_{2,i}\mid \groundset'_{i,l}$ for all $i\in [\beta]$ and $l\geq 0$. We apply \Cref{lem:integerweightsparsifier} to these to obtain w-DCS $D'_{i,l}\subseteq \groundset'_{i,l}$ for all $i\in [\beta]$ and $l\geq 0$. Our candidate w-DCS for $\matroid'_{1,i}$ and $\matroid'_{2,i}$ will be $D'_i:=\cup_{l\geq 0}D'_{i,l}$. We now show each of the relevant properties.
    \begin{enumerate}
    \item {\bf(Sparseness)} From \Cref{lem:integerweightsparsifier} we can conclude that $|D'_i|=O(\gamma_{\varepsilon}\cdot r\cdot \log W)$ for all $i\in [\beta]$.
   \item {\bf(Approximation Ratio)} From \Cref{lem:integerweightsparsifier}, we know that $D'_{i,l}$ contains a $(\frac{2}{3}-\varepsilon)$-approximation to the maximum weight common independent set of $\matroid'_{1,i}\mid \groundset'_{i,l}$ and $\matroid'_{2,i}\mid \groundset'_{i,l}$. Consequently, by \Cref{cor:greedymergeguarantee}, we can conclude that $D'_i:=\cup_{l\geq 0}D'_{i,l}$ contains a $(\frac{2}{3}-\varepsilon)$-approximate maximum weight common independent set of $\matroid'_{1,i}$ and $\matroid'_{2,i}$. 
   \item {\bf(Robustness)} Let $D'_{i,l,A}$ be the w-DCS of $\matroid'_{1,i}\mid \groundset'_{i,l}\cap \groundset'_A$ and $\matroid'_{2,i}\mid \groundset'_{i,l}\cap \groundset'_A$ and let $D'_{i,A}=\cup_{l\geq 0}D_{i,l,A}$. We want to show that $\matroid'_{1,i}\mid D'_{i,A}\cup (\groundset'_{i}\cap \groundset'_B)$ and $\matroid'_{2,i}\mid D'_{i,A}\cup (\groundset'_{i}\cap \groundset'_B)$ contain a $(\frac{2}{3}-\varepsilon)$-approximation to the maximum weight independent set of $\matroid'_{1,i}$ and $\matroid'_{2,i}$. To see this, observe that by \Cref{lem:integerweightsparsifier} we know that $D'_{i,l,A}$ is a w-DCS and therefore, $\matroid'_{1,i}\mid D'_{i,l,A}\cup (\groundset'_{i,l}\cap \groundset'_B)$ and $\matroid'_{2,i}\mid D'_{i,l,A}\cup (\groundset'_{i,l}\cap \groundset'_B)$ contain a $(\frac{2}{3}-\varepsilon)$-approximation to the maximum weight independent set of $\matroid'_{1,i}\mid \groundset'_{i,l}$ and $\matroid'_{2,i}\mid \groundset'_{i,l}$. Thus, by \Cref{cor:greedymergeguarantee} we can conclude that $\matroid'_{1,i}\mid D'_{i,A}\cup (\groundset'_{i}\cap \groundset'_B)$ and $\matroid'_{2,i}\mid D'_{i,A}\cup (\groundset'_{i}\cap \groundset'_B)$ contain a $(\frac{2}{3}-\varepsilon)$-approximation to the maximum weight independent set of $\matroid'_{1,i}$ and $\matroid'_{2,i}$.
    \end{enumerate}
    Finally, since at least one of the $\beta$ pairs of matroids among $\matroid'_{i,1}$ and $\matroid'_{i,2}$ contain a $(1-\varepsilon)$-approximate maximum weight common independent set, we can conclude that the w-DCS $D'_i$ corresponding to that particular pair satisfies all the three properties with respect to $\matroid'_1$ and $\matroid'_2$.
\end{proof}

\subsubsection{Application to Random Order Streams.}

Our unfolding technique and aspect-ratio reduction also give an algorithm for approximate maximum weight matroid intersection in random-order streams, matching the state-of-the-art for the cardinality problem. The algorithm however does not follow via the reduction we described so far, instead it follows via more white-box methods. First, we state the result of \cite{HuangS24} for unweighted matroid intersection in random-order streams.

\begin{lemma}[\cite{HuangS24}]\label{lem:huangrandomorder}
    One can extract from a randomly-ordered stream of elements a common independent set in two matroids with an approximation ratio of $(\frac{2}{3}-\varepsilon)$ using space $O(\frac{r\cdot \log n\cdot \log \min\{r_1,r_2\}}{\varepsilon^3})$.
\end{lemma}

We need the following definition to describe the algorithm.

\begin{definition}[Contraction]
For a set $S \subseteq \mathcal{N}$, the contracted matroid $\mathcal{M} / S$ is $ (\mathcal{N} \setminus S, \mathcal{I}/S)$ where $\mathcal{I}/S := \{I \subseteq \mathcal{N} \setminus S | I \cup B_S \in \mathcal{I}\}$ for a base $B_S$ of $\mathcal{M}|S$.    
\end{definition}

Our main result in this section is the following result.

\begin{lemma}\label{lem:randomlyorderedweighted}
    One can extract from a randomly-ordered stream of elements a common independent set which is a $(\frac{2}{3}-\varepsilon)$-approximation to the maximum weight independent set using space $O(r\cdot \log n\cdot \log \min\{r_1,r_2\}\cdot \log W\cdot \gamma_{\varepsilon})$, where $\opt$ is the rank of the common independent set.
\end{lemma}

In order to obtain this result, we show the following reduction which is inspired by the work of \cite{HashemiW24}, who prove the same result for matchings. Their starting point is the following model.

\begin{definition}[$b$-batch model]
    Let $\matroid_1=(\groundset,\independentsets_1)$, $\matroid_2=(\groundset,\independentsets_2)$ be an instance of matroid intersection. An adversary partitions $\groundset$ into batches $\groundset_1,\cdots, \groundset_q$ of size at most $b$ each. The arrival order of the batches $\{\groundset_{i_1},\cdots, \groundset_{i_q}\}$ is chosen uniformly at random. The elements in each batch arrive consecutively. 
\end{definition}

Our goal will be to show the following lemma which implies \Cref{lem:randomlyorderedweighted}.

\begin{lemma}[$b$-batch model to weighted random order]\label{lem:bbatchtoweighted}
    If there exists an algorithm $\mathcal{A}_u$ for unweighted matroid intersection in the $b$-batch model that computes an $\alpha$-approximate maximum cardinality common independent set using space $S(n,r,b)$, then there exists an algorithm $\mathcal{A}_w$ for the weighted matroid intersection in the random-order model that computes a $(1-\varepsilon)\alpha$-approximation to the maximum weight common independent set using space $S(n\gamma_{\varepsilon},r\gamma_{\varepsilon},\gamma_{\varepsilon})\log W$.
\end{lemma}
\begin{proof}[Proof Sketch]
    We will show that it is sufficient to construct an algorithm $\mathcal{A}_{W}$ for random-order streams that given a weighted matroid intersection with weights in $[W]$ computes an $\alpha$-approximate maximum weight common independent set using space $S(n,r,W)$. By \Cref{lem:streamingaspectratioreduction} and \Cref{remark:generalityreduction}, we can get the requisite $\mathcal{A}_w$ from $\mathcal{A}_W$. 
    We now proceed to describe $\mathcal{A}_W$. When an element $e\in \groundset'$ arrives, we unfold it to create the batch of elements $e_{1},\cdots, e_{W}$ and feed these to $\mathcal{A}_u$. Thus, $\mathcal{A}_u$ is applied to the unfolded versions of the matroid with batches of size $W$. Thus, $\mathcal{A}_u$ computes an $\alpha$-approximate maximum cardinality common independent set $I$ of the unfolded matroids in space $S(nW,rW,W)$. The algorithm $\mathcal{A}_W$ then at the end of the stream computes an $\alpha$-approximate maximum weight common independent set $I'$ of the original matroids by refolding $I$ (this is implied by \Cref{lem:refoldguarantee}).
\end{proof}

Thus, our goal in this subsection is to prove the following lemma.

\begin{lemma}[Algorithm for $b$-batch model]\label{lem:algbbatch}
    Given an instance of unweighted matroid intersection, there exists an algorithm which computes a $(\frac{2}{3}-\varepsilon)$-approximation to the maximum cardinality independent set in the $b$-batch model using space $O(\frac{\opt\cdot b\cdot \log n\cdot \log (\min\{r_1,r_2\})}{\varepsilon^3})$.
\end{lemma}

The authors in \cite{HuangS24} use DCS to show \Cref{lem:huangrandomorder}. We state some definitions relevant to DCS and the algorithm. Then, we will give a modified version of their algorithm and main lemmas and derive a proof sketch of \Cref{lem:algbbatch}.

\begin{definition}[Density, Definition 2 of \cite{HuangS24}]
    Let $\mathcal{M}=(\groundset, \independentsets)$ be a matroid. The density of a subset $\groundset'\subseteq \groundset$ is defined as
    \begin{align*}
        \rho_{\matroid}(\groundset')=\frac{\card{\groundset'}}{\rk_{\matroid}(\groundset')}.
    \end{align*}
\end{definition}

\begin{definition}[Density-based Decomposition, Algorithm 1 of \cite{HuangS24}]
   Let $\matroid=(\groundset, \independentsets)$ be a matroid and let $\groundset'\subseteq \groundset$, $\matroid'=\matroid\mid \groundset'$. The density based decomposition of $\matroid'$ is defined as follows. For $j=1,\cdots, \rk(\matroid)$ do the following:
    \begin{enumerate}
        \item $U_j\leftarrow$ the densest subset of largest cardinality in $\matroid/\cup_{l=1}^{j-1} U_{l}$.
    \end{enumerate}
\end{definition}

\begin{definition}[Associated Density of an Element]
   Let $\matroid=(\groundset,\independentsets)$ be a matroid. Let $\groundset'\subseteq \groundset$ and suppose $\matroid'=\matroid |\groundset'$ and suppose $U_1,\cdots, U_{\rk}$ is a density based decomposition of $\matroid'$. Then, given an element $e\in \groundset$, its associated density $\tilde{\rho}_{\matroid}(e)$ is defined as
   \begin{align*}
   \begin{split}
       \tilde{\rho}_{\matroid}(e)=\begin{cases}
           \rho_{\matroid'/\cup_{i=1}^{j-1} U_i}(U_j)&\text{ for }j=\min\{j\in \set{1,\cdots, \rk(\matroid)}: e\in \text{span}_{\matroid}(\cup_{i=1}^j U_i)\}\\&\text{ if }e\in \text{span}_{\matroid}(\groundset')\\
           0&\text{ otherwise.}
       \end{cases}
    \end{split}
   \end{align*}
\end{definition}
\begin{definition}[Density Constrained Subset]\label{def:DCS}
    Let $\matroid_1=(\groundset,\independentsets_1)$ and $\matroid_2=(\groundset,\independentsets_2)$ be two matroids. Let $\beta, \beta^{-}$ be integers such that $\beta\geq \beta^-+7$. A subset $\groundset'\subseteq \groundset$ is called a $(\beta,\beta^-)$-DCS if the following two properties hold.
    \begin{enumerate}
        \item\label{property:DCSone} For any $e\in \groundset'$, $\tilde{\rho}_{\matroid_1}(e)+\tilde{\rho}_{\matroid_2}(e)\leq \beta$, 
        \item\label{property:DCStwo} For any $e\in \groundset\setminus\groundset'$, $\tilde{\rho}_{\matroid_1}(e)+\tilde{\rho}_{\matroid_2}(e)\geq \beta^-$. 
    \end{enumerate}
\end{definition}

\begin{lemma}[Properties of DCS]
We have the following implications of the properties of the density constrained subset.
\begin{enumerate}
\item    For any set $\groundset'\subset\groundset$ satisfying \Cref{def:DCS}(\ref{property:DCSone}), $\card{\groundset'}\leq \beta\cdot \opt$.
\item Let $\varepsilon>0$, $\beta, \beta^-$ be integers such that $\beta\geq \beta^-+7$ and $(\beta^--4)\cdot (1+\varepsilon)\geq \beta$. Any $(\beta,\beta^-)$-DCS contains a $(\frac{2}{3}-\varepsilon)$-approximation to the maximum cardinality independent set. 
    \end{enumerate}
\end{lemma}

The main structural tool the algorithm uses is as follows.

\begin{lemma}\label{lem:boundeddensityandunderfull}
    We say that $\groundset'\subseteq \groundset$ has bounded density iff $\tilde{\rho}_{\matroid_1}(e)+\tilde{\rho}_{\matroid_2}(e)\leq \beta$. Suppose $\groundset'$ is such a set. Suppose $X\subseteq \groundset\setminus \groundset'$ consist of all elements $e\in \groundset\setminus \groundset'$ that have $\tilde{\rho}_{\matroid_1}(e)+\tilde{\rho}_{\matroid_2}(e)<\beta$ (referred to as $(\groundset',\beta^-,\beta)$-underfull elements), then $\matroid_1\mid \groundset'\cup X$ and $\matroid_2\mid \groundset'\cup X$ contain a $(\frac{2}{3}-\varepsilon)$-approximation to $\opt$. 
\end{lemma}

We now present the modified version of Algorithm 2 in \cite{HuangS24}.

    \begin{algorithm}
	\caption{Algorithm for computing an intersection of two matroids in a $b$-batched model}\label{algo:intersection-streaming}
	\begin{algorithmic}[1]
	\State $\groundset' \gets \emptyset$
	\State $\forall\,0 \leq i \leq \log_2 k,\, \alpha_i \gets \left\lfloor\frac{\varepsilon \cdot n}{\log_2(k) \cdot (2^{i+2}\beta^2 + 1)}\right\rfloor$
    \Procedure{Phase 1}{}
	\For{$i = 0 \dots \log_2 k$}
	    \State $\textsc{ProcessStopped} \gets \textsc{False}$
	    \For{$2^{i+2}\beta^2 + 1$ iterations}
	        \State $\textsc{FoundUnderfull} \gets \textsc{False}$
	        \For{$\frac{\alpha_i}{b}$ batches}
	            \State let $e$ be the next element in the stream
	            \If{$\tilde{\rho}_{\mathcal{M}_1}(v) + \tilde{\rho}_{\mathcal{M}_2}(v) < \beta^-$}
	                \State add $e$ to $\groundset'$
	                \State $\textsc{FoundUnderfull} \gets \textsc{True}$
                    \While{there exists $e' \in \groundset' : \tilde{\rho}_{\mathcal{M}_1}(e') + \tilde{\rho}_{\mathcal{M}_2}(e') > \beta$}
                        \State remove $e'$ from $\groundset'$
                    \EndWhile
	            \EndIf
	        \EndFor
	        \If{$\textsc{FoundUnderfull} = \textsc{False}$}
	           \State Go to Phase 2
	        \EndIf
	    \EndFor
%	    \If{$\textsc{ProcessStopped} = \textsc{True}$}
%	        \State \textbf{break} from the loop
%	    \EndIf
	\EndFor
\EndProcedure
\Procedure{Phase 2}{}
	\State $X \gets \emptyset$
	\For{each $e$ remaining element in the stream}
	    \If{$\tilde{\rho}_{\mathcal{M}_1}(v) + \tilde{\rho}_{\mathcal{M}_1}(v) < \beta^-$}
	        \State add $e$ to $X$
	    \EndIf
	\EndFor
	\State \Return the maximum common independent set in $\groundset' \cup X$
\EndProcedure
	\end{algorithmic}
	\end{algorithm}

We now show the modified version of intermediate lemmas of \cite{HuangS24}. Since the proofs follow in a relatively straightforward manner, we omit them.

\begin{lemma}[Modified version of \protect{\cite[Claim 5.1]{HuangS24}}]
    Let $E^{late}$ denote the elements that appear in Phase 2. Then, $\matroid_1\mid E^{late}$ and $\matroid_2\mid E^{late}$ contain a $(1-\varepsilon)$-approximation to the maximum cardinality independent set with probability at least $1-\frac{1}{2^{\varepsilon^2\opt}}$.
\end{lemma}
\begin{proof}[Proof Sketch]
    In \Cref{algo:intersection-streaming}, Phase 1 of an algorithm lasts for the first $\frac{\varepsilon n}{b}$ randomly chosen batches. Consider any arbitrary maximum cardinality common independent set $I$, each batch can contain at most $b$ elements from $I$. In expectation at most $\frac{\varepsilon|I|}{b}$ of these batches arrive in Phase 1. By a standard application of concentration bounds, one can argue this is also true with high probability.
\end{proof}

We additionally have the following properties of \Cref{algo:intersection-streaming} which follow from the analogous statement in \cite[Theorem 5.1]{HuangS24}.

\begin{property2}[Properties of \Cref{algo:intersection-streaming}]\label{lem:propalgorithm}
    The following properties hold for \Cref{algo:intersection-streaming}.
    \begin{enumerate}
        \item Phase 1 terminates within first $\frac{\varepsilon\cdot n}{b}$ batches of the stream.
        \item The set $\groundset'$ has bounded density.
        \item The total number of $(\groundset',\beta,\beta^-)$-underfull elements in $\groundset^{late}$ is at most $O(b\cdot \opt\cdot \log n\cdot\log k\cdot \beta^2\cdot\varepsilon^{-1})$.
    \end{enumerate}
\end{property2}

We will now prove \Cref{lem:algbbatch}.

\begin{proof}[Proof of \Cref{lem:algbbatch}]
    From \Cref{lem:propalgorithm} we can conclude that $\groundset'$ has bounded density. Additionally, let $X$ denote the number of $(\groundset',\beta,\beta^-)$-underfull elements in $\groundset^{late}$ is at most $O(b\cdot\opt\cdot\log n\cdot \log k\cdot\beta^2\cdot \varepsilon^{-1})$. By \Cref{lem:boundeddensityandunderfull} we can conclude that the maximum common independent set of $\matroid_1\mid \groundset'\cup X$ and $\matroid_2\mid\groundset'\cup X$ contains a $(\frac{2}{3}-\varepsilon)$-approximation to the maximum cardinality independent set of $\matroid_1,\matroid_2$. By \Cref{lem:propalgorithm} we can conclude that the number of elements stored in the memory is at most $O(b\cdot \opt\cdot \log n\cdot\log k\cdot \beta^2\cdot\varepsilon^{-1})$.
\end{proof}

Finally, we conclude this section by proving our main result.

\begin{proof}[Proof of \Cref{lem:randomlyorderedweighted}]
This follows immediately from \Cref{lem:algbbatch} and \Cref{lem:bbatchtoweighted}.  
\end{proof}

\subsection{Applications to Bipartite Weighted $b$-Matching}
\label{sec:bmatching}

In addition to these settings, specific instances of matroid intersection, namely, bipartite $b$-matching have also been studied in models such as MPC, parallel shared-memory work-depth model, distributed blackboard model, and streaming model. Our structural theorems (namely, aspect ratio reduction, matroid unfolding, refolding) also apply to $b$-matching. Using these techniques, we obtain some new results for bipartite weighted $b$-matching in these models, or match the state-of-the-art. We briefly define the $b$-matching problem, each of the above models, and summarize our results. 

\begin{definition}[$b$-matching]
Let $G=(V,E)$ be a graph, $b\in \mathbb{N}^{V}_{\geq 1}$ a vector and $M$ a set of edges. We say that $M$ is a $b$-matching if for each $v\in V$ there are at most $b_v$ edges in $M$ incident to $v$.
\end{definition}

\subsubsection{Massively Parallel Computation (MPC) Model.}
When considering a graph problem $G=(V,E)$ in the MPC model \cite{BeameKS17,GoodrichSZ11,KarloffSV10}, we consider $M$ machines, and each machine knows only a part of the graph (that is, only some of the vertices and edges). A parameter of concern in this model is $S$, the local memory per machine. Since the combined memory of the machines should be able to hold the graph, it must be that $M\cdot S\geq |V|+|E|$ and it is common that one assumes $M\cdot S=\tilde{O}(|V|+|E|)$. We refer to the combined memory as the global memory. In this model, initially the graph is partitioned among the machines. The communication graph between the machines is a complete graph, and the communication proceeds in synchronous rounds. In a round, each machine performs some (polynomial-time) computation on its memory contents. After this, the machines send messages to each other. The only restriction is that the total amount of information it sends or receives cannot be more than its memory. In this paper, we consider the {\bf sublinear} memory MPC: In this case, $S=n^{1-c}$ for some positive constant $c>0$.

\subsubsection{Parallel shared-memory work-depth model.}

This is a parallel model \cite{JaJa1992} where different machines/processors can read and write from the same shared-memory and process instructions in parallel. The computational measures of interest are {\bf work}, which is the total computation performed and {\bf depth}, which is the longest sequence of dependent instructions performed by the algorithm. 

\subsubsection{Distributed Blackboard Model.} According to \cite{DobzinskiNO19}, this model features $n$ players, each representing a vertex in a bipartite graph. These players interact through a multi-round communication protocol with a central coordinator that maintains a shared ``blackboard''. In every round, each player sends a message to the coordinator using their own private information and the contents of the blackboard. The coordinator then sends a message (not necessarily identical) to each player. The process ends when the coordinator, based on the information accumulated on the blackboard, decides to terminate and produces the final $b$-matching. The model efficiency is evaluated by the total number of communication rounds and the bit-length of the messages sent each round.

\subsubsection{Our Results.} In this section we briefly summarize our results for $b$-matching in the above mentioned models.
\begin{enumerate}
\item In the work-depth model, we are able to extend the result of \cite{LiuKK23} for unweighted $b$-matching to the weighted setting. In particular we obtain an algorithm with total work $O(\gamma_{\varepsilon}\cdot m\cdot \log n)$ and depth $O(\frac{\log^3n\log W}{\varepsilon^2})$ that computes a $(1-\varepsilon)$-approximate maximum weight bipartite b-matching. To the best of our knowledge, this is the first result for the weighted problem in this model. 
\item In the streaming model, we are able to extend the result of \cite{LiuKK23} for unweighted bipartite $(1-\varepsilon)$-approximate  $b$-matching to the weighted setting. In particular, we obtain a $O(\frac{1}{\varepsilon^2})$-pass, $O(\gamma_{\varepsilon}\cdot (\sum_{i\in L}b_i+|R|))$ space. Prior to this work, the algorithms of \cite{AhnG18,Quanrud24} took at least $\log n$ passes.  
\item In the distributed blackboard model, we are able to extend the result of \cite{LiuKK23} for unweighted $b$-matching to the weighted setting. In particular, we obtain a protocol which has total bit complexity of $O(n\max_{i\in L}b_i\cdot \gamma_{\varepsilon}\cdot \log W)$. 
\item In the MPC model, for the case of the sublinear memory regime, \cite{LackiMRS25} studied the allocation problem (where only one side of the graph has capacities greater than 1). Here, for the unweighted version, they obtained an $O_{\varepsilon}(\sqrt{\log \lambda})$ round algorithm for $(1-\varepsilon)$-approximation, using $O(\lambda n)$ total memory and sublinear local memory (here, $\lambda$ is the arboricity). We extend this to the case of weighted allocation, with the algorithm taking $O_{\varepsilon}(\sqrt{\log \lambda})$ rounds, and at a $\gamma_{\varepsilon}$ factor loss in the local and total memory\footnote{It is easy to see that the arboricity of the unfolded graph only increases by a factor of $\gamma_{\varepsilon}$.} but only being able to compute the weight of a $(1-\epsilon)$-approximate allocation.
\end{enumerate}

\begin{remark}The reason our algorithms are not able to go beyond the value version in the sublinear regime MPC model is because in this model, we don’t know how to compute a $(1-\varepsilon)$-approximate weighted $b$-matching in sparse graphs.
\end{remark}

\section{Future Work}
\label{sec:futurework}
We obtained a reduction which converts an $\alpha$-approximate unweighted matroid intersection algorithm into a weighted one with an approximation ratio of $\alpha(1-\varepsilon)$. The reduction is versatile but incurs a loss of a factor of $\varepsilon^{-O(\varepsilon^{-1})}$ in various parameters pertaining to the models of interest. An additional feature of our reduction is its \emph{non-adaptivity}, which means that in each weight class, one can pick an independent set arbitrarily and independently of other weight classes.\footnote{We recall that in our reduction, we split up the elements of the matroid into different weight classes. In each weight class, we can run the unfolding and unweighted algorithms in parallel. Therefore, we refer to our reduction as a \emph{non-adaptive} reduction from weighted to unweighted matroid intersection.} On one hand, non-adaptivity implies that in the streaming and one-way communication complexity settings we obtain a lossless (in terms of the number of rounds/passes) reduction. On the other hand, a result of \cite{BernsteinCDLST25} (Claim C.1) suggests that even for bipartite matching, one cannot have a weighted-to-unweighted reduction that \begin{inparaenum}[(a)]
    \item is non-adaptive, \item works for arbitrary approximation ratios, and \item has $\poly(\varepsilon^{-1})$ aspect ratio. 
\end{inparaenum}
Indeed, the work of \cite{ChekuriQ16}, which gives a scaling algorithm for weighted matroid intersection by using an unweighted algorithm in each scale, to the best of our knowledge, relies on \begin{inparaenum}[(i)]
    \item the unweighted algorithms computing a maximal set of short augmenting paths in the exchange graph, \item\label{item:adaptivity} is adaptive as independent sets derived in the subsequent scales depend on previous ones, and \item is only applicable in $(1-\varepsilon)$-approximation ratio regime. 
\end{inparaenum} The last two properties affect the portability of this technique to one-pass streaming and one-way communication complexity settings. This is because (\ref{item:adaptivity}) leads to an increase in the number of passes (\cite{LiuKK23,HuangS22}). Moreover, there are results which show that it is impossible to achieve an approximation ratio of $(1-\varepsilon)$ in these models \cite{Kapralov21,GoelKK12}. This leaves us with the following open question: 
\begin{mdframed}[backgroundcolor=lightgray!40,topline=false,rightline=false,leftline=false,bottomline=false,innertopmargin=2pt]
Building on the techniques of \cite{BernsteinCDLST25}, can we obtain a non-adaptive weighted-to-unweighted reduction for matroid intersection, with an aspect ratio of $\poly(\varepsilon^{-1})$, in the $(1-\varepsilon)$-approximate regime?
\end{mdframed} 

\section*{Acknowledgements}
We are grateful to Sebastian Forster for his valuable feedback and comments during this project. We thank the anonymous reviewers for their helpful feedback.

\bibliographystyle{splncs04}
\bibliography{references}

@inproceedings{ChakrabartyLSSW19,
  title={Faster matroid intersection},
  author={Chakrabarty, Deeparnab and Lee, Yin Tat and Sidford, Aaron and Singla, Sahil and Wong, Sam Chiu-wai},
  booktitle={2019 IEEE 60th Annual Symposium on Foundations of Computer Science (FOCS)},
  pages={1146--1168},
  year={2019},
  organization={IEEE}
}

@inproceedings{GargJS21,
  author       = {Paritosh Garg and
                  Linus Jordan and
                  Ola Svensson},
  title        = {Semi-streaming Algorithms for Submodular Matroid Intersection},
  booktitle    = {Integer Programming and Combinatorial Optimization - 22nd International
                  Conference ({IPCO})},
                  year = {2021}
}

@article{Terao24,
  author       = {Tatsuya Terao},
  title        = {Deterministic (2/3-{$\epsilon$})-Approximation of Matroid Intersection
                  using Nearly-Linear Independence-Oracle Queries},
  journal      = {CoRR},
  volume       = {abs/2410.18820},
  year         = {2024}
}

@inproceedings{BlikstadT25,
  author       = {Joakim Blikstad and
                  Ta{-}Wei Tu},
  title        = {Efficient Matroid Intersection via a Batch-Update Auction Algorithm},
  booktitle    = {2025 Symposium on Simplicity in Algorithms ({SOSA})},
  pages        = {226--237},
  year = {2025}}

@inproceedings{HuangS24,
  author       = {Chien{-}Chung Huang and
                  Fran{\c{c}}ois Sellier},
  title        = {Robust Sparsification for Matroid Intersection with Applications},
  booktitle    = {Proceedings of the 2024 {ACM-SIAM} Symposium on Discrete Algorithms},
  pages        = {2916--2940},
  year = {2024}
}

@inproceedings{HashemiW24,
  author       = {Diba Hashemi and
                  Weronika Wrzos{-}Kaminska},
  title        = {Weighted Matching in the Random-Order Streaming and Robust Communication
                  Models},
  booktitle    = {Approximation, Randomization, and Combinatorial Optimization. Algorithms
                  and Techniques ({APPROX/RANDOM})},
  pages        = {16:1--16:26},
  year         = {2024}
}

@inproceedings{GuptaP13,
  author       = {Manoj Gupta and
                  Richard Peng},
  title        = {Fully Dynamic {(1+} e)-Approximate Matchings},
  booktitle    = {54th Annual {IEEE} Symposium on Foundations of Computer Science ({FOCS})},
  pages        = {548--557},
  year         = {2013}
}

@inproceedings{BernsteinCDLST25,
  author       = {Aaron Bernstein and
                  Jiale Chen and
                  Aditi Dudeja and
                  Zachary Langley and
                  Aaron Sidford and
                  Ta{-}Wei Tu},
  title        = {Matching Composition and Efficient Weight Reduction in Dynamic Matching},
  booktitle    = {Proceedings of the 2025 Annual {ACM-SIAM} Symposium on Discrete Algorithms ({SODA})},
  pages        = {2991--3028},
  year         = {2025}
}

@inproceedings{HuangKK16,
  author       = {Chien{-}Chung Huang and
                  Naonori Kakimura and
                  Naoyuki Kamiyama},
  title        = {Exact and Approximation Algorithms for Weighted Matroid Intersection},
  booktitle    = {Proceedings of the Twenty-Seventh Annual {ACM-SIAM} Symposium on Discrete
                  Algorithms ({SODA})},
  pages        = {430--444},
  year         = {2016}
}

@inproceedings{Edmonds03,
  title={Submodular functions, matroids, and certain polyhedra},
  author={Edmonds, Jack},
  booktitle={Combinatorial Optimization—Eureka, You Shrink! Papers Dedicated to Jack Edmonds 5th International Workshop Aussois, France, March 5--9, 2001},
  pages={11--26},
  year={2003}}

@inproceedings{Kapralov21,
  author       = {Michael Kapralov},
  title        = {Space Lower Bounds for Approximating Maximum Matching in the Edge
                  Arrival Model},
  booktitle    = {Proceedings of the 2021 {ACM-SIAM} Symposium on Discrete Algorithms ({SODA})},
  pages        = {1874--1893},
  year         = {2021}
}

@inproceedings{BlikstadBMN21,
  title={Breaking the quadratic barrier for matroid intersection},
  author={Blikstad, Joakim and Van Den Brand, Jan and Mukhopadhyay, Sagnik and Nanongkai, Danupon},
  booktitle={Proceedings of the 53rd Annual ACM SIGACT Symposium on Theory of Computing (STOC)},
  pages={421--432},
  year={2021}
}

@article{Lawler75,
  title={Matroid intersection algorithms},
  author={Lawler, Eugene L},
  journal={Mathematical programming},
  volume={9},
  number={1},
  pages={31--56},
  year={1975}
}

@article{Nguyen19,
  title={A note on Cunningham's algorithm for matroid intersection},
  author={Nguyen, Huy L},
  journal={arXiv preprint arXiv:1904.04129},
  year={2019}
}

@article{AignerD71,
  title={Matching theory for combinatorial geometries},
  author={Aigner, Martin and Dowling, Thomas A},
  journal={Transactions of the American Mathematical Society},
  volume={158},
  number={1},
  pages={231--245},
  year={1971}
}

@inproceedings{BernsteinDL21,
  author       = {Aaron Bernstein and
                  Aditi Dudeja and
                  Zachary Langley},
  title        = {A framework for dynamic matching in weighted graphs},
  booktitle    = {Proceedings of the 53rd Annual {ACM} {SIGACT} Symposium on Theory of Computing (STOC)},
  pages        = {668--681},
  year         = {2021}
}

@book{Schrijver03,
  title={Combinatorial optimization: polyhedra and efficiency},
  author={Schrijver, Alexander},
  publisher={Springer},
  volume={24},
  year={2003}
}

@article{BeameKS17,
  author       = {Paul Beame and
                  Paraschos Koutris and
                  Dan Suciu},
  title        = {Communication Steps for Parallel Query Processing},
  journal      = {J. {ACM}},
  volume       = {64},
  number       = {6},
  pages        = {40:1--40:58},
  year         = {2017}
}

@book{JaJa1992,
  title={An Introduction to Parallel algorithms},
  author={J{\'a}J{\'a}, Joseph},
  year={1992},
  publisher = {Addison Wesley Longman Publishing Co., Inc.},
address = {USA}
}

@inproceedings{HuangS22,
  author       = {Chien{-}Chung Huang and
                  Fran{\c{c}}ois Sellier},
  title        = {Maximum Weight b-Matchings in Random-Order Streams},
  booktitle    = {30th Annual European Symposium on Algorithms ({ESA})},
  pages        = {68:1--68:14},
year = {2022}
}

@article{DobzinskiNO19,
  author       = {Shahar Dobzinski and
                  Noam Nisan and
                  Sigal Oren},
  title        = {Economic efficiency requires interaction},
  journal      = {Games Econ. Behav.},
  volume       = {118},
  pages        = {589--608},
  year         = {2019}
}

@article{AhnG18,
  title={Access to data and number of iterations: Dual primal algorithms for maximum matching under resource constraints},
  author={Ahn, Kook Jin and Guha, Sudipto},
  journal={ACM Transactions on Parallel Computing (TOPC)},
  volume={4},
  number={4},
  pages={1--40},
  year={2018}
}

@inproceedings{KarloffSV10,
  author       = {Howard J. Karloff and
                  Siddharth Suri and
                  Sergei Vassilvitskii},
  title        = {A Model of Computation for MapReduce},
  booktitle    = {Proceedings of the Twenty-First Annual {ACM-SIAM} Symposium on Discrete
                  Algorithms ({SODA})},
  pages        = {938--948},
  year         = {2010}
}

@inproceedings{GoodrichSZ11,
  author       = {Michael T. Goodrich and
                  Nodari Sitchinava and
                  Qin Zhang},
  title        = {Sorting, Searching, and Simulation in the MapReduce Framework},
  booktitle    = {Algorithms and Computation - 22nd International Symposium ({ISAAC})},
  pages        = {374--383},
  year = {2011}
  }

@inproceedings{LackiMRS25,
  author       = {Jakub Lacki and
                  Slobodan Mitrovic and
                  Srikkanth Ramachandran and
                  Wen{-}Horng Sheu},
  title        = {Faster {MPC} Algorithms for Approximate Allocation in Uniformly Sparse
                  Graphs},
  booktitle    = {Proceedings of the 37th {ACM} Symposium on Parallelism in Algorithms
                  and Architectures (SPAA)},
  pages        = {339--349},
  year         = {2025}
}

@article{AssadiB21b,
  title={On the robust communication complexity of bipartite matching},
  author={Assadi, Sepehr and Behnezhad, Soheil},
  journal={Approximation, Randomization, and Combinatorial Optimization. Algorithms and Techniques (APPROX/RANDOM)},
  year={2021}
}

@inproceedings{CrouchS14,
  author       = {Michael S. Crouch and
                  Daniel M. Stubbs},
  title        = {Improved Streaming Algorithms for Weighted Matching, via Unweighted
                  Matching},
  booktitle    = {Approximation, Randomization, and Combinatorial Optimization. Algorithms
                  and Techniques ({APPROX/RANDOM})},
  series       = {LIPIcs},
  volume       = {28},
  pages        = {96--104},
  year         = {2014}
}

@inproceedings{Quanrud24,
  author       = {Kent Quanrud},
  title        = {Adaptive Sparsification for Matroid Intersection},
  booktitle    = {51st International Colloquium on Automata, Languages, and Programming
                  ({ICALP})},
  pages        = {118:1--118:20},
  year         = {2024}
}

@inproceedings{ChakrabartyLS0W19,
  author       = {Deeparnab Chakrabarty and
                  Yin Tat Lee and
                  Aaron Sidford and
                  Sahil Singla and
                  Sam Chiu{-}wai Wong},
  title        = {Faster Matroid Intersection},
  booktitle    = {60th {IEEE} Annual Symposium on Foundations of Computer Science, ({FOCS})},
  pages        = {1146--1168},
  year         = {2019}
}

@inproceedings{Blikstad21,
  author       = {Joakim Blikstad},
  title        = {Breaking O(nr) for Matroid Intersection},
  booktitle    = {48th International Colloquium on Automata, Languages, and Programming
                  (ICALP)},
  pages        = {31:1--31:17},
  year = {2021}
}

@article{Cunningham86,
  title={Improved bounds for matroid partition and intersection algorithms},
  author={Cunningham, William H},
  journal={SIAM Journal on Computing},
  volume={15},
  number={4},
  pages={948--957},
  year={1986}
}

@inproceedings{ChekuriQ16,
  author       = {Chandra Chekuri and
                  Kent Quanrud},
  title        = {A Fast Approximation for Maximum Weight Matroid Intersection},
  booktitle    = {Proceedings of the Twenty-Seventh Annual {ACM-SIAM} Symposium on Discrete
                  Algorithms (SODA)},
  pages        = {445--457},
  year         = {2016}
}

@article{Frank81,
title = {A weighted matroid intersection algorithm},
journal = {Journal of Algorithms},
volume = {2},
number = {4},
pages = {328-336},
year = {1981},
issn = {0196-6774},
author = {András Frank}
}

@article{KaoLST01,
  author       = {Ming{-}Yang Kao and
                  Tak Wah Lam and
                  Wing{-}Kin Sung and
                  Hing{-}Fung Ting},
  title        = {A Decomposition Theorem for Maximum Weight Bipartite Matchings},
  journal      = {{SIAM} J. Comput.},
  volume       = {31},
  number       = {1},
  pages        = {18--26},
  year         = {2001}
}

@inproceedings{Assadi24,
  title={A Simple (1- $\varepsilon$)-Approximation Semi-Streaming Algorithm for Maximum (Weighted) Matching},
  author={Assadi, Sepehr},
  booktitle={7th SIAM Symposium on Simplicity in Algorithms (SOSA)},
  pages={337--354},
  year={2024}
}

@inproceedings{GoelKK12,
  title={On the communication and streaming complexity of maximum bipartite matching},
  author={Goel, Ashish and Kapralov, Michael and Khanna, Sanjeev},
  booktitle={Proceedings of the Twenty-Third Annual {ACM-SIAM} Symposium on Discrete
                  Algorithms (SODA)},
  pages={468--485},
  year={2012}
}

@InProceedings{assadiB21,
  author =	{Assadi, Sepehr and Behnezhad, Soheil},
  title =	{{Beating Two-Thirds For Random-Order Streaming Matching}},
  booktitle =	{48th International Colloquium on Automata, Languages, and Programming (ICALP)},
  pages =	{19:1--19:13},
  ISBN =	{978-3-95977-195-5},
  ISSN =	{1868-8969},
  year =	{2021}
  
}

@InProceedings{AssadiB18,
  author =	{Assadi, Sepehr and Bernstein, Aaron},
  title =	{{Towards a Unified Theory of Sparsification for Matching Problems}},
  booktitle =	{2nd Symposium on Simplicity in Algorithms (SOSA)},
  pages =	{11:1--11:20},
  ISBN =	{978-3-95977-099-6},
  ISSN =	{2190-6807},
  year =	{2019}
}

@inproceedings{HuangS23,
  title={(1-eps)-Approximate Maximum Weighted Matching in poly (1/eps, log n) Time in the Distributed and Parallel Settings},
  author={Huang, Shang-En and Su, Hsin-Hao},
  booktitle={ACM Symposium on Principles of Distributed Computing},
  year={2023}
}

@inproceedings{LiuKK23,
  author       = {Quanquan C. Liu and
                  Yiduo Ke and
                  Samir Khuller},
  title        = {Scalable Auction Algorithms for Bipartite Maximum Matching Problems},
  booktitle    = {Approximation, Randomization, and Combinatorial Optimization. Algorithms
                  and Techniques (APPROX/RANDOM)},
  pages        = {28:1--28:24},
  year = {2023}
}

@article{AzarmehrB23,
  title={Robust communication complexity of matching: EDCS achieves 5/6 approximation},
  author={Azarmehr, Amir and Behnezhad, Soheil},
  journal={arXiv preprint arXiv:2305.01070},
  year={2023}
}

@article{LeeSV13,
author = {Lee, Jon and Sviridenko, Maxim and Vondr\'{a}k, Jan},
title = {Matroid Matching: The Power of Local Search},
journal = {SIAM Journal on Computing},
volume = {42},
number = {1},
pages = {357-379},
year = {2013}
}
\appendix

\section{Rounding Arbitrary Weights to Integers}
\label{sec:rescaling}

Suppose we are given as input matroids $\matroid'_1=(\groundset',\independentsets'_1,w)$ and $\matroid'_1=(\groundset',\independentsets'_1,w)$ with $w:\groundset'\rightarrow \mathbb{R}^{> 0}$ with aspect ratio $R:=\frac{\max_{e\in \groundset'}w(e)}{\min_{e\in \groundset'}w(e)}$. In this section, we are concerned with showing that via rescaling and rounding, we can construct matroids $\matroid^R_{1}=(\groundset',\independentsets'_1,w_r)$ and $\matroid^R_2=(\groundset',\independentsets'_2,w_r)$ such that $w_r(e)\in \set{1,2,\cdots, W}$ where $W$ is an integer and $W\leq \lceil\frac{2R}{\varepsilon}\rceil$. We now describe our transformation.

\begin{itemize}
    \item {\bf Rescaling} We obtain the intermediate weight function $w_s(e)=\frac{2w(e)}{\varepsilon\cdot W_{\min}}$, where $W_{\min}=\min_{e\in \groundset'}w(e)$.
    \item {\bf Rounding} We bucket $e\in \groundset'$ in buckets with geometrically increasing boundaries. Thus, if $w_s(e)\in [(1+\varepsilon)^i,(1+\varepsilon)^{i+1})$ (the $i$th bucket), then we let $w_r(e)=\lfloor(1+\varepsilon)^i\rfloor$ be the rounded weights.
\end{itemize}

We now show that the rounded weights are not much different from the scaled weights.

\begin{lemma}\label{lem:scalingandrounding}
    For all $e\in \groundset'$, $w_r(e)\geq \frac{w_s(e)}{(1+\varepsilon)^2}$ and $w_r(e)\leq w_s(e)$.
\end{lemma}
\begin{proof}
    We first show that for all relevant values of $i$, $(1+\varepsilon)^i-1\geq \frac{(1+\varepsilon)^{i+1}}{(1+\varepsilon)^2}$. Subsequently, we will show that the lemma holds. After moving around some terms, we can conclude that the first equation holds provided $(1+\varepsilon)^{i}\geq \frac{1}{\varepsilon\cdot(1+\varepsilon)}$. 
    
    Now, consider $w_s(e)$ for $e\in\groundset'$. Observe that $w_{s}(e)\geq \frac{2}{\varepsilon}$. Thus, for each $e\in \groundset'$, we can conclude that if $e$ belongs in the $i_e$th bucket, then $(1+\varepsilon)^{i_e}\geq \frac{1}{\varepsilon\cdot(1+\varepsilon)}$. Thus, we have,
    \begin{align*}
        w_{r}(e)\geq (1+\varepsilon)^{i_e}-1 \geq (1+\varepsilon)^{i_e-1}\geq w_s(e)\cdot (1+\varepsilon)^{-2}. \qedhere
    \end{align*}
\end{proof}

\begin{lemma}
   Let $S$ is an $\alpha$-approximate common independent set of $\matroid^R_1,\matroid^R_2$ then $S$ is also a $(1-2\varepsilon)\cdot \alpha$-approximate maximum weight independent set of $\matroid'_1,\matroid'_2$. 
\end{lemma}
\begin{proof}
 Let $S^*$ and $S^r$ be a maximum weight common independent set with respect to $w$ and $w_r$ respectively. For the common independent set $S$ we have that $w_r(S) \geq \alpha w_r(S^r)$. Since, $S^*$ is also a common independent set, we get $w_r(S) \geq \alpha w_r(S^*) \geq \alpha (1+\varepsilon)^{-2} w_s(S^*)$ by \Cref{lem:scalingandrounding}. 
 
 Further, we have $w^r(S) \geq \alpha \frac{2}{\varepsilon\cdot W_{\min}}(1+\varepsilon)^{-2} w(S^*)$ by the definition of $w_s$. Finally, if $\varepsilon \leq \frac{1}{2}$, we get
  \begin{equation*}
       w(S) \geq \frac{W_{\min}\varepsilon}{2}w_s(S) \geq  \frac{W_{\min}\varepsilon}{2}w_r(S)\geq \frac{W_{\min}\varepsilon}{2}\left(\alpha \frac{2}{\varepsilon\cdot W_{\min}}(1+\varepsilon)^{-2}\right) w(S^*) 
    \end{equation*}
The rightmost side can be rewritten as follows
\begin{equation*}
    \alpha (1+\varepsilon)^{-2}w(S^*) \geq \alpha (1-2\varepsilon)w(S^*).
\end{equation*}
This concludes the proof.
\end{proof}
\end{document}